\documentclass[11pt]{article}

\usepackage{url}
\usepackage{graphicx}
\usepackage{latexsym}
\usepackage{fullpage}
\usepackage{float}
\usepackage{amssymb,amsmath,amsthm,amsfonts}
\usepackage{xspace}
\usepackage[usenames]{color}
\usepackage{dsfont,mathrsfs}

\allowdisplaybreaks[3]      

\newlength\abovesectionskip
\newlength\belowsectionskip
\def\sectionfont{\normalfont\Large\bfseries}
\newlength\abovesubsectionskip
\newlength\belowsubsectionskip
\def\subsectionfont{\normalfont\large\bfseries}
\newlength\abovesubsubsectionskip
\newlength\belowsubsubsectionskip
\def\subsubsectionfont{\normalfont\normalsize\bfseries}
\newlength\aboveparagraphskip
\newlength\belowparagraphskip
\def\paragraphfont{\normalfont\normalsize\bfseries}

\makeatletter
\def\section{\@startsection{section}{1}{\z@}{-\abovesectionskip}%
               {\belowsectionskip}{\sectionfont}}
\def\subsection{\@startsection{subsection}{2}{\z@}{-\abovesubsectionskip}%
                  {\belowsubsectionskip}{\subsectionfont}}
\def\subsubsection{\@startsection{subsubsection}{3}{\z@}%
                    {-\abovesubsubsectionskip}{\belowsubsubsectionskip}%
                    {\subsubsectionfont}}
\def\paragraph{\@startsection{paragraph}{4}{\z@}{-\aboveparagraphskip}%
                 {-\belowparagraphskip}{\paragraphfont}}
\makeatother

\newcommand{\bits}{\{-1,1\}}

\newcommand{\T}{\mathcal{T}}
\newcommand{\depth}{{\mathrm{depth}}}
\newcommand{\D}{\mathcal{D}}
\newcommand{\U}{\mathcal{U}}
\renewcommand{\Pr}{\mathbf{Pr}}

\newcommand{\eqdef}{\stackrel{\textrm{def}}{=}}

\newcommand{\maxcut}{\textsc{MaxCut}}
\newcommand{\eigensum}{\Upsilon}
\newcommand{\tO}{\widetilde{O}}

\newcommand{\normal}{\mathcal{N}}
\newcommand{\tr}{\textrm{tr}}
\newcommand{\reg}{R}

\newcommand{\smallball}{T}
\newcommand{\largerball}{S}

\newcommand{\lambdamin}{\lambda_{\mathrm{min}}}
\newcommand{\lambdamax}[1]{\| #1 \|_\infty}
\newcommand{\Inf}{\mathrm{Inf}}

\DeclareMathSymbol{\qedsymb} {\mathord}{AMSa}{"04}

\newcommand{\R}{\mathbb{R}}
\newcommand{\N}{\mathbb{N}}
\newcommand{\sgn}{\mathrm{sgn}}
\newcommand{\E}{\mathbf{E}}
\newcommand{\Var}{\mathbf{Var}}
\newcommand{\convolution}{*}
\newcommand{\eps}{\varepsilon}
\newcommand{\poly}{\mathrm{poly}}

\newcommand{\TODO}[1]{}

\newcommand{\inprod}[1]{\left\langle #1 \right\rangle}
\newcommand{\ceil}[1]{\left\lceil #1 \right\rceil}
\newcommand{\floor}[1]{\left\lfloor #1 \right\rfloor}

\newcommand{\thmabove}{3pt}
\newcommand{\thmbelow}{3pt}

    \newtheoremstyle{mythmstyle}
      {\thmabove}   
      {\thmbelow}   
      {}            
      {}            
      {\bfseries}   
      {. }          
      {2.5pt}       
      {\thmname{#1}\thmnumber{ #2}\thmnote{ \normalfont (#3)}}   
    \theoremstyle{mythmstyle}
    \newtheorem{theorem}{Theorem}[section]\numberwithin{equation}{section}
    \newtheorem{corollary}[theorem]{Corollary}
    \newtheorem{definition}[theorem]{Definition}

    \newtheorem{fact}[theorem]{Fact}
    \newtheorem{claim}[theorem]{Claim}
    \newtheorem{lemma}[theorem]{Lemma}
    \newtheorem{remark}[theorem]{Remark}

\newcommand{\ClaimName}[1]{\label{clm:#1}}
\newcommand{\CorollaryName}[1]{\label{cor:#1}}
\newcommand{\DefinitionName}[1]{\label{def:#1}}
\newcommand{\EquationName}[1]{\label{eq:#1}}
\newcommand{\FactName}[1]{\label{fact:#1}}
\newcommand{\LemmaName}[1]{\label{lem:#1}}
\newcommand{\RemarkName}[1]{\label{rem:#1}}

\newcommand{\SectionName}[1]{\label{sec:#1}}
\newcommand{\TheoremName}[1]{\label{thm:#1}}

\newcommand{\Claim}[1]{Claim~\ref{clm:#1}}
\newcommand{\Corollary}[1]{Corollary~\ref{cor:#1}}

\newcommand{\Equation}[1]{Eq.\:\eqref{eq:#1}}
\newcommand{\Fact}[1]{Fact~\ref{fact:#1}}
\newcommand{\Lemma}[1]{Lemma~\ref{lem:#1}}

\newcommand{\Section}[1]{Section~\ref{sec:#1}}
\newcommand{\Theorem}[1]{Theorem~\ref{thm:#1}}

\newcommand{\BoldCorollary}[1]{\textbf{Corollary~\ref{cor:#1} (restatement).}}

\newcommand{\BoldLemma}[1]{\textbf{Lemma~\ref{lem:#1} (restatement).}}

\newcommand{\BoldTheorem}[1]{\textbf{Theorem~\ref{thm:#1} (restatement).}}

\renewenvironment{proof}{\noindent\textbf{Proof: }\ignorespaces}
  {\hspace*{\fill}$\Box$\medskip}
\newcommand{\proofbelow}{3pt}
\newcommand{\afterproof}{\hfill $\blacksquare$ \par \vspace{\proofbelow}}
\newcommand{\aftersubproof}{\hfill $\Box$ \par \vspace{\proofbelow}}
\renewenvironment{proof}{\noindent\textbf{Proof.}\,}{\afterproof}

\newenvironment{proofof}[1]{\noindent\textbf{Proof} \,(of
  #1).\,}{\afterproof}
\newenvironment{proofof-sketch}[1]{\noindent\textbf{Proof (Sketch).}
  \,(of #1).\,}{\afterproof}

\makeatletter
\renewcommand{\section}{\@startsection{section}{1}{0pt}{-12pt}{5pt}{\large\bf}}
\renewcommand{\subsection}{\@startsection{subsection}{2}{0pt}{-12pt}{-5pt}{\normalsize\bf}}
\renewcommand{\subsubsection}{\@startsection{subsubsection}{3}{0pt}{-12pt}{-5pt}{\normalsize\bf}}
\makeatother

\title{Bounded Independence Fools Degree-$2$ Threshold Functions}

\author{Ilias
  Diakonikolas\footnotemark[2]\\\texttt{ilias@cs.columbia.edu} \and
  Daniel
  M. Kane\footnotemark[3]\\\texttt{dankane@math.harvard.edu}\and
  Jelani Nelson\footnotemark[4]\\\texttt{minilek@mit.edu}}

\date{}

\begin{document}

\footnotetext[1]{Department of Computer Science, Columbia University.
Research supported by NSF grant CCF-0728736, and by an Alexander
S. Onassis Foundation Fellowship. Part of this work was done
while visiting IBM Almaden.}
\footnotetext[1]{Harvard University, Department of
  Mathematics. Supported by a
  National Defense Science and Engineering Graduate (NDSEG)
  Fellowship.}
\footnotetext[3]{MIT Computer Science and Artificial Intelligence
  Laboratory. Supported by a National
  Defense Science and Engineering Graduate (NDSEG) Fellowship, and in
  part by the Center for Massive Data Algorithmics
  (MADALGO) - a center of the Danish National Research
  Foundation. Part of
  this work was done while visiting IBM Almaden.}

\maketitle

\setcounter{page}{0}

\begin{abstract}
\thispagestyle{empty} \noindent Let $x$ be a random vector coming from
any $k$-wise independent distribution over $\{-1,1\}^n$. For an
$n$-variate degree-$2$ polynomial $p$, we prove that $\E[\sgn(p(x))]$
is determined up to an additive $\eps$ for $k = \poly(1/\eps)$. This
answers an open question of Diakonikolas et al. (FOCS 2009). Using
standard constructions of $k$-wise independent distributions, we
obtain a broad class of explicit generators that $\eps$-fool the class
of degree-$2$ threshold functions with seed length $\log n \cdot
\poly(1/\eps)$.

Our approach is quite robust: it easily extends to yield that the
intersection of any constant number of degree-$2$ threshold functions
is $\eps$-fooled by $\poly(1/\eps)$-wise independence. Our results
also hold if the entries of $x$ are $k$-wise independent standard
normals, implying for example that bounded independence derandomizes
the Goemans-Williamson hyperplane rounding scheme.

To achieve our results, we introduce a technique we dub {\em
  multivariate FT-mollification}, a generalization of the univariate
form introduced by Kane et al. (SODA 2010) in the context of streaming
algorithms.  Along the way we prove a generalized hypercontractive
inequality for quadratic forms which takes the operator norm of the
associated matrix into account. These techniques may be of independent
interest.
\end{abstract}

\newpage

\section{Introduction}\SectionName{intro}

This paper is concerned with the power of limited independence to fool low-degree polynomial threshold functions. A degree-$d$
\emph{polynomial threshold function} (henceforth PTF), is a boolean function $f:\bits^n \to \bits$  expressible as $f(x) =
\sgn(p(x))$, where $p$ is an $n$-variate degree-$d$ polynomial with real coefficients, and $\sgn$ is $-1$ for negative
arguments and $1$ otherwise. PTFs have played an important role in computer science since the early perceptron work of Minsky
and Papert \cite{MinskyPapert:69}, and have since been extensively investigated in circuit complexity and communication
complexity~\cite{ABF+:94, Beigel:94a, Bruck:90,BruckSmolensky:92,GHR:92,
  HMP+:93,KrausePudlak:98,Nisan:94,
  OdonnellServedio:08jcss,RazborovSherstov:08,Saks:93}, learning
theory \cite{KOS04,KS04,Sherstov09}, and more.

A distribution $\D$ on $\bits^n$ is said to $\eps$-fool a function
$f:\bits^n \to \bits$ if
\[|\E_{x \sim \D}[f(x)] - \E_{x\sim \U}[f(x)]| \leq \eps \]
where $\U$ is the uniform distribution on $\bits^n$. A distribution $\D$ on $\bits^n$ is $k$-wise independent if every
restriction of $\D$ to $k$ coordinates is uniform on $\bits^k$. 
Despite their simplicity,  $k$-wise independent distributions have
 been a
surprisingly powerful and versatile derandomization tool, fooling  complex
functions such as AC$^0$ circuits
\cite{Bazzi07,Razborov09,Braverman09} and half-spaces \cite{DGJSV09}.
As a result, this class of distributions has played a fundamental role
in many areas of theoretical computer science.

\smallskip

\noindent \textbf{Our Results.} The problem we study is the following: How large must $k = k(n,d,\eps)$ be in order for
\emph{every} $k$-wise independent distribution on $\bits^n$ to $\eps$-fool the class of degree-$d$ PTF's? The $d=1$ case of
this problem was recently considered in~\cite{DGJSV09}, where it was shown that $k(n,1,\eps) = \widetilde{\Theta}(1/\eps^2)$,
independent of $n$, with an alternative proof to much of the argument given in \cite{KNW10}. The main open problem
in~\cite{DGJSV09} was to identify $k = k(n,d,\eps)$ for $d\geq2$. In this work, we make progress on this question by proving
the following:

\begin{theorem} \TheoremName{main}
Any $\tilde{\Omega}(\eps^{-9})$-wise independent distribution on
$\bits^n$ $\eps$-fools all degree-$2$
PTFs.
\end{theorem}

Prior to this work, no nontrivial result was known for $d>1$; it was not even known whether $o(n)$-wise independence suffices
for constant $\eps$. Using known constructions of $k$-wise independent distributions \cite{AlonBI86,ChG89}, \Theorem{main}
gives a large class of pseudo-random generators (PRGs) for degree-$2$ PTFs with seed length $\log (n) \cdot \tO(\eps^{-9})$.

Our techniques are quite robust. Our approach yields for example that \Theorem{main} holds not only over the hypercube, but
also over the $n$-variate Gaussian distribution. This already implies that the Goemans-Williamson hyperplane rounding scheme
\cite{GW95} (henceforth ``GW rounding'') can be derandomized using $\poly(1/\eps)$-wise independence\footnote{We note that
other derandomizations of GW rounding are known with better dependence on $\eps$, though not solely using $k$-wise
independence;
  see \cite{MahajanH95,Sivakumar02}.}. Our technique also readily
extends to show that the intersection of $m$ halfspaces,
or even $m$ degree-$2$ threshold functions, is $\eps$-fooled by
$\poly(1/\eps)$-wise independence for any constant $m$ (over
both the hypercube and the multivariate Gaussian). One consequence of
this is that $O(1/\eps^2)$-wise independence suffices
for GW rounding.

Another consequence of \Theorem{main} is that bounded independence suffices for the invariance principle of Mossell, O'Donnell,
and Oleszkiewicz in the case degree-$2$ polynomials. Let $p(x)$ be an $n$-variate degree-$2$ multi-linear polynomial with ``low
influences''. The invariance principle roughly says that the distribution of $p$ is essentially invariant if $x$ is drawn from
the uniform distribution on $\bits^n$ versus the standard $n$-dimensional Gaussian distribution $\normal(0,1)^n$. Our result implies
that the $x$'s do not need to be fully independent for the invariance principle to apply, but that bounded independence
suffices.

\smallskip
\noindent \textbf{Motivation and Related Work.} The literature is rich with explicit generators for various natural classes of
functions. Recently, there has been much interest in not only constructing PRGs for natural complexity classes, but also in
doing so with as broad and natural a family of PRGs as possible. One example is the recent work of Bazzi
\cite{Bazzi07} on fooling depth-$2$ circuits (simplified by Razborov \cite{Razborov09}), and of Braverman \cite{Braverman09} on
fooling AC$^0$, with bounded independence\footnote{Note that a PRG for AC$^0$ with qualitatively similar -- in fact slightly
better -- seed length had being already given by Nisan \cite{Nisan91}.}.

Simultaneously and independently from our work, Meka and Zuckerman \cite{MZ10}
constructed PRGs against degree-$d$ PTFs with seed length $\log
n \cdot 2^{O(d)}\cdot (1/\eps)^{8d+3}$ \cite{MZ10}. That is, their
seed length for $d=2$ is similar to ours (though worse by a
$\poly(1/\eps)$ factor). 
However, their result is incomparable to ours since their pseudorandom
generator is customized for PTFs, and not based on $k$-wise
independence alone.
We
believe that the ideas in our proof may lead to generators with better
seed-length\footnote{An easy probabilistic argument
shows that there exists PRGs for degree-$d$ PTFs with seed-length
$O(d\log(n/\eps))$.}, and that some of the techniques
we introduce are of independent interest.

In other recent and independent works, \cite{GOWZ10,HKM10} give PRGs
for intersections
of $m$ halfspaces (though not degree-$2$ threshold
functions). The
former has polynomial dependence on $m$ and requires only
bounded independence as well (and considers other functions of halfspaces
beside intersections), while the latter has poly-logarithmic
dependence
on $m$ under the Gaussian measure but is not solely via bounded
independence. Our
dependence on $m$ is polynomial.

\section{Notation}\SectionName{notation}
Let $p: \{-1,1\}^n \to \R$ be a polynomial and $p(x) =
\sum_{S\subseteq [n]} \widehat{p}_S \chi_S$ be its Fourier-Walsh
expansion, where $\chi_S(x) \eqdef \prod_{i\in S} x_i$. The \emph{influence} of variable $i$ on $p$ is $\Inf_i(p) \eqdef
\sum_{S\ni i} \widehat{p}_S^2$, and the \emph{total influence} of $p$ is $\Inf(p) = \sum_{i=1}^n \Inf_i(p).$ If $\Inf_i(p)\leq
\tau \cdot \Inf(p)$ for all $i$, we say that the polynomial $p$ is \emph{$\tau$-regular}.  If $f(x)=\sgn(p(x))$, where $p$
is $\tau$-regular, we say that $f$ is a \emph{$\tau$-regular} PTF.

For $R \subseteq \R^d$ denote by $I_R:\R^d \to \{0,1\}$ its characteristic function. It will be convenient in some of the
proofs to phrase our results in terms of $\eps$-fooling $\E[I_{[0,
  \infty)}(p(x))]$ as opposed to $\E[\sgn(p(x))]$. It is
straightforward that these are equivalent up to changing $\eps$ by a
factor of $2$.

We frequently use $A\approx_{\eps} B$ to denote that
$|A-B| = O(\eps)$, and we let the function $d_2(x,R)$ denote the $L_2$ distance from some $x\in\R^d$ to a region $R\subseteq \R^d$.

\section{Overview of our proof of
  \Theorem{main}}\SectionName{multivariate-ft}
The program of our proof follows the outline of the proof in~\cite{DGJSV09}. We first prove that bounded independence fools the
class of \emph{regular} degree-$2$ PTF's. We then reduce the general case to the regular case to show that bounded independence
fools all degree-$2$ PTF's. The bulk of our proof is to establish the first step; this is the most challenging part of this
work and where our main technical contribution lies. The second step is achieved by adapting the recent results of
\cite{DSTW09}.

We now elaborate on the first step. Let $f:\bits^n \to \bits$ be a boolean function. To show that $f$ is fooled by $k$-wise
independence, it suffices -- and is in fact necessary -- to prove the existence of two degree-$k$ ``sandwiching'' polynomials
$q_u, q_l:\bits^n \to \bits$ that approximate $f$ in a certain technical sense (see e.g. \cite{Bazzi07, BGGP:07}). Even though
this is an $n$-dimensional approximation problem, it may be possible to exploit the additional structure of the function under
consideration to reduce it to a low-dimensional problem. This is exactly what is done in both \cite{DGJSV09} and \cite{KNW10}
for the case of regular halfspaces.

We now briefly explain the approaches of \cite{DGJSV09} and
\cite{KNW10}. Let $f(x) = \sgn(\langle w, x\rangle)$ be an
$\eps^2$-regular
halfspace, i.e. $\|w\|_2 =1$ and $\max_i |w_i| \leq \eps$. An
insight used in \cite{DGJSV09} (and reused in
\cite{KNW10}) is the
following: the random variable $\langle w, x \rangle$ behaves
approximately like a standard Gaussian, hence it can be treated
as if it was one-dimensional. Thus, both \cite{DGJSV09} and
\cite{KNW10} construct a (different in each case) univariate
polynomial $P: \R \to \R$ that is a ``good'' approximation to the sign
function under the normal distribution in $\R$ (in the case of
\cite{KNW10}, the main point of the alternative proof was to avoid
explicitly reasoning about any such polynomials, but the existence of
such a polynomial is still implicit in the proof).
The
desired $n$-variate sandwiching polynomials are then obtained
(roughly) by setting $q_u(x) = P(\langle w, x \rangle)$ and
$q_u(x) = -P(-\langle w, x \rangle)$. It turns out that this approach suffices for the case of halfspaces. In \cite{DGJSV09}
the polynomial $P$ is constructed using approximation theory
arguments. In \cite{KNW10} it is obtained by taking a truncated
Taylor expansion of a certain smooth approximation to the sign
function, constructed via a method dubbed ``Fourier
Transform mollification'' (henceforth FT-mollification). We elaborate
in \Section{ft} below.

Let $f(x) = \sgn(p(x))$ be a regular degree-$2$ PTF. A first natural attempt to handle this case would be to use the univariate
polynomial $P$ described above -- potentially allowing its degree to increase -- and then take $q_u(x) = P(p(x))$, as before.
Unfortunately, such an approach fails for both constructions outlined above. We elaborate on this issue in
\Section{previous-fail}.

\subsection{FT-mollification}\SectionName{ft}
FT-mollification is a general procedure to obtain a smooth function
with bounded derivatives that approximates some bounded
function $f$. The univariate version of the method in the context of
derandomization was introduced in
\cite{KNW10}. In this paper we generalize it to the
multivariate setting and later use it to prove our main theorem.

For the univariate case, where $f:\R\rightarrow\R$, \cite{KNW10}
defined $\tilde{f}^c(x)= (c\cdot
\hat{b}(c\cdot t)\convolution f(t))(x)$ for a parameter $c$, where
$\hat{b}$ has unit integral and is the Fourier
transform of a smooth function $b$ of compact support (a so-called
{\em bump function}). Here
``$\convolution$'' denotes convolution.
The idea of smoothing functions via convolution with a smooth
approximation of the Dirac delta function is old, dating back to
``Friedrichs mollifiers''
\cite{Friedrichs44} in 1944. Indeed, the only difference between
Friedrichs mollification and FT-mollification is that in the
former, one convolves $f$ with the scaled bump function, and not its
Fourier transform.  The switch to the
Fourier transform is made to have better control on the high-order
derivatives of the resulting smooth function, which is
crucial for our application.

In our context, the method can be illustrated as follows. Let $X = \sum_i a_iX_i$ for independent $X_i$. Suppose we would like
to argue that $\E[f(X)] \approx_{\eps} \E[f(Y)]$, where $Y = \sum_i a_iY_i$ for $k$-wise independent $Y_i$'s that are
individually distributed as the $X_i$. Let $\tilde{f}^c$ be the FT-mollified version of $f$. If the parameter $c=c(\eps)$ is
appropriately selected, we can guarantee that $|f(x)-\tilde{f}^c(x)|<\eps$ ``almost everywhere'', and furthermore have ``good''
upper bounds on the high-order derivatives of $\tilde{f}^c$. We could then hope to show the following chain of inequalities:
$\E[f(X)] \approx_{\eps} \E[\tilde{f}^c(X)] \approx_{\eps} \E[\tilde{f}^c(Y)] \approx_{\eps} \E[f(Y)].$ To justify the first
inequality, note $f$ and $\tilde{f}^c$ are close almost everywhere, and so it suffices to argue that $X$ is sufficiently
anti-concentrated in the small region where they are not close. The second inequality would use Taylor's theorem, bounding the
error via upper bounds on moment expectations  of $X$ and the high-order derivatives of $\tilde{f}^c$. Showing the final
inequality would be similar to the first, except that one needs to justify that even under $k$-wise independence the
distribution of $Y$ is sufficiently anti-concentrated. We note that the argument outlined above was used in \cite{KNW10} to
provide an alternative proof that bounded independence fools regular halfspaces, and to optimally derandomize Indyk's moment
estimation algorithm in data streams \cite{Indyk06}. However, this univariate approach fails for degree-$2$ PTFs for technical
reasons (see \Section{previous-fail}).

We now describe our switch to multivariate FT-mollification. Let
$f:\{-1,1\}^n\rightarrow \{-1,1\}$ be arbitrary
and let $S\subset\R^n$ with $f^{-1}(1)\subseteq S\subseteq
\R^n\backslash f^{-1}(-1)$. Then fooling $\E[f(x)]$ and
fooling $\E[I_S(x)]$ are
equivalent. A natural attempt to this end would be to generalize
FT-mollification to $n$ dimensions, then FT-mollify $I_S$ and argue as
above using the multivariate Taylor's theorem. Such an
approach is perfectly valid, but as one might expect, there is a penalty for working over high dimensions. Both our
quantitative bounds on the error introduced by FT-mollifying, and the
error coming from the multivariate Taylor's theorem,
increase with the dimension. Our approach is then to
find a \emph{low-dimensional representation} of such a
region $S$ which allows us to obtain the desired bounds. We elaborate below on how this can be accomplished in our setting.

\subsection{Our Approach} Let $f = \sgn(p)$ be a regular multilinear degree-$2$ PTF with $\|p\|_2=1$ (wlog). Let
us assume for simplicity that $p$ is a quadratic form; handling the
additive linear form and constant is easy. The first conceptual step in our
proof is this: we decompose $p$ as $p_1-p_2+p_3$, where $p_1,p_2$ are positive semidefinite quadratic forms with no small
non-zero eigenvalues and $p_3$ is indefinite with all eigenvalues small in magnitude. This decomposition, whose existence
follows from elementary linear algebra, is particularly convenient for
the following reason: for $p_1,p_2$, we are able to exploit their
positive semidefiniteness to obtain better bounds from Taylor's
theorem, and for $p_3$ we can
establish moment bounds that are \emph{strictly} stronger than the
ones that follow from hypercontractivity for general
quadratic forms (our \Theorem{eigenbound}, which may be of independent
interest).  The fact that we need $p_1,p_2$ to not only be positive
semidefinite, but to also have no small eigenvalues, arises for technical
reasons; specifically, quadratic forms with no small non-zero
eigenvalues satisfy much better tail bounds, which plays a role in our
analysis.

We now proceed to describe the second conceptual step of the proof, which involves multivariate FT-mollification. As suggested
by the aforementioned, we would like to identify a region $R \subseteq \R^3$ such that $I_{[0,\infty)}(p(x))$ can be written
as $I_R(F(x))$ for some $F:\bits^n\rightarrow\R^3$ that
depends on the $p_i$, then FT-mollify $I_R$. The region $R$ is
selected as follows: note
we can write $p_3(x) = x^TA_{p_3}x$, where $A_{p_3}$ is a real
symmetric matrix with trace $\eigensum$. We consider the
region $R = \{x : x_1^2 - x_2^2 + x_3 + \eigensum \ge 0\} \subseteq \R^3$. Observe that $I_{[0,\infty)}(p(x)) =
I_{\reg}(\sqrt{p_1(x)}, \sqrt{p_2(x)}, p_3(x) - \eigensum)$. (Recall that $p_1, p_2$ are positive-semidefinite, hence the first
two coordinates are always real.) We then prove via FT-mollification that $\E[I_R(\sqrt{p_1(x)}, \sqrt{p_2(x)}, p_3(x) -
\eigensum)]$ is preserved within $\eps$ by bounded independence.

The high-level argument is of similar flavor as the one outlined above for the case of halfspaces, but the
details are more elaborate. The proof makes essential use of good tail
bounds for $p_1,p_2$, a new
moment bound for $p_3$,
properties of
FT-mollification, and a variety of other tools such as the Invariance
Principle \cite{MOO10} and the anti-concentration bounds
of \cite{CW01}.

\smallskip

\noindent \textbf{Organization.} \Section{thisisit} contains the results we will need on multivariate FT-mollification. In
\Section{spectral-moment} we give our improved moment bound on
quadratic forms. \Section{regular} contains the analysis of the
regular case, and \Section{main-thm} concludes the proof of our main
theorem. \Section{intersections} summarizes our results on
intersections.

\section{Multivariate FT-mollification}\SectionName{thisisit}
\begin{definition}
In {\em hyperspherical coordinates} in $\R^d$, we represent a point
$x = (x_1,\ldots,x_d)$ by $x_i = r\cos(\phi_i)\prod_{j=1}^{i-1}
\sin(\phi_j)$ for $i<d$, and $x_d = r\prod_{j=1}^{d-1} \sin(\phi_j)$. Here
$r = \|x\|_2$ and the $\phi_i$ satisfy $0\le \phi_i \le \pi$
for $i<d-1$, and $0\le \phi_{d-1} < 2\pi$.
\end{definition}

\begin{fact}
Let $J$ be the Jacobian matrix corresponding to the change of
variables from Cartesian to hyperspherical coordinates.  Then 
$$\det(J) = r^{d-1}\prod_{i=1}^{d-2}\sin^{d-1-i}(\phi_i) .$$
\end{fact}

We define the bump function $b:\R^d\rightarrow\R$ by
$$ b(x) = \sqrt{C_d}\cdot \begin{cases} 1-\|x\|_2^2 \ &
  \textrm{for} \ \|x\|_2
  < 1 \\ 0 \ &\textrm{otherwise} \end{cases} .$$
The value $C_d$ is chosen so that $\|b\|_2^2 = 1$. 
We note that $b$ is not smooth (its mixed partials do not exist
at the boundary of the unit sphere), but we will only ever need that
$\frac{\partial}{\partial x_i} b\in L^2(\R^d)$ for all $i\in[d]$.

Henceforth, we make
the setting
$$ A_d = C_d \cdot \int_0^{2\pi} \int_{[0,\pi]^{d-2}}
\left(\prod_{i=1}^{d-2}\sin^{d-1-i}(\phi_i)\right) d\phi_1
d\phi_2 \cdots d\phi_{d-1} .$$

We let $\hat{b}:\R^d\rightarrow\R$ denote the Fourier transform of
$b$, i.e.
$$\hat{b}(t) = \frac{1}{(\sqrt{2\pi})^d}\int_{\R^d} b(x)e^{-i\inprod{x,t}}dx
.$$
Finally, $B:\R^d\rightarrow\R$ denotes the function
$\hat{b}^2$, and we define $B_c:\R^d\rightarrow\R$ by
$$ B_c(x_1,\ldots,x_d) = c^d\cdot B(c x_1,\ldots,c x_d) .$$

\begin{definition}[Multivariate FT-mollification]
For $F:\R^d\rightarrow \R$ and given $c>0$, we define the {\em
  FT-mollification} $\tilde{F}^c:\R^d\rightarrow\R$ by
$$ \tilde{F}^c(x) = (B_c \convolution F)(x) = \int_{\R^d}
B_c(y)F(x-y)dy .$$
\end{definition}

In this section we give several quantitative properties of
FT-mollification.  We start off with a few lemmas that will be useful
later.

\begin{lemma}\LemmaName{unit-integral}
For any $c>0$,
$$ \int_{\R^d} B_c(x)dx = 1 .$$
\end{lemma}
\begin{proof}
Since $B = \hat{b}^2$, the stated integral when $c=1$ is
$\|\hat{b}\|_2^2$, which
is $\|b\|_2^2 = 1$ by Plancherel's theorem.
For
general $c$, make the
change of variables $u = (cx_1,\ldots,cx_d)$ then integrate over $u$.
\end{proof}

Before presenting the next lemma, we familiarize the reader with 
some multi-index notation. A $d$-dimensional multi-index is a vector $\beta
\in \N^d$ (here $\N$ is the nonnegative integers).  For $\alpha,\beta
\in \N^d$, we say $\alpha \le \beta$ if the inequality holds
coordinate-wise, and for such $\alpha,\beta$ we define $|\beta| =
\sum_i \beta_i$, 
$\binom{\beta}{\alpha} = \prod_{i=1}^d \binom{\beta_i}{\alpha_i}$, and
$\beta! = \prod_{i=1}^d \beta_i!$. 
For $x\in\R^d$ we use $x^\beta$ to denote $\prod_{i=1}^d
x_i^{\beta_i}$, and for $f:\R^d\rightarrow\R$ we use
$\partial^\beta f$ to denote $\frac{\partial^{|\beta|}}{\partial
  x_1^{\beta_1} \cdots \partial x_d^{\beta_d}} f$.

\begin{lemma}\LemmaName{l1bound}
For any $\beta\in\N^d$, $
\|\partial^\beta B\|_1 \le 2^{|\beta|}$.
\end{lemma}
\begin{proof}
We have
\begin{equation*}
\partial^\beta B = \sum_{\alpha\le \beta} \binom{\beta}{\alpha}
\left(\partial^\alpha
\hat{b}\right)\cdot
\left(\partial^{\beta-\alpha}\hat{b} \right)
\end{equation*}

Thus,
\begin{eqnarray}
\nonumber \left\|\partial^\beta B\right\|_1 &=&
\left\|\sum_{\alpha\le
    \beta}\binom{\beta}{\alpha}\left(\partial^\alpha \hat{b}\right)
  \cdot
\left(\partial^{\beta-\alpha}\hat{b}\right)\right\|_1\\
&\le&
\sum_{\alpha\le \beta}\binom{\beta}{\alpha}
\left\|\partial^\alpha \hat{b}\right\|_2 \cdot
\left\|\partial^{\beta-\alpha}
  \hat{b}\right\|_2\EquationName{awesome-cs}\\
&=&
\sum_{\alpha\le \beta}\binom{\beta}{\alpha}
\left\|x^{\alpha}\cdot b\right\|_2 \cdot
\left\|x^{\beta-\alpha}\cdot b\right\|_2\EquationName{ft-preservel2}\\
&\le&
\sum_{\alpha\le
  \beta}\binom{\beta}{\alpha}
\EquationName{l2atmost1}\\
&=&
2^{|\beta|}\EquationName{combinatorially-nice}
\end{eqnarray}
\Equation{awesome-cs} follows by
Cauchy-Schwarz. \Equation{ft-preservel2} follows from Plancherel's
theorem, since the Fourier transform of
$\partial^\alpha \hat{b}$ is
$x^{\alpha}\cdot b$, up to factors of $i$. \Equation{l2atmost1}
follows since $\|x^\alpha\cdot b\|_2 \le \|b\|_2 = 1$.
\Equation{combinatorially-nice} is
seen combinatorially.  Suppose we have $2d$ buckets
$A_i^j$ for $(i,j)\in [d]\times [2]$. We also have $|\beta|$ balls,
with each having one of
$d$ types with $\beta_i$ balls of type $i$. Then the number of ways to
place balls into buckets such that balls of type $i$ only go into
some $A_i^j$ is $2^{|\beta|}$ (each ball has $2$ choices).
However, it is also $\sum_{\alpha\le \beta}\binom{\beta}{\alpha}$,
since for every placement of balls we must place some
number $\alpha_i$ balls of type $i$ in $A_i^1$ and $\beta_i-\alpha_i$
balls in $A_i^2$.
\end{proof}

\begin{lemma}
Let $z>0$ be arbitrary.  Then
$$\int_{\|x\|_2 > dz} B(x)dx = O(1/z^2) .$$
\end{lemma}
\begin{proof}
Consider the integral
$$S = \int_{\R^d} \|x\|_2^2\cdot B(x)dx = \sum_{i=1}^d
\left(\int_{\R^d} x_i^2\cdot B(x)dx\right) .$$
Recalling that $B = \hat{b}^2$, the
Fourier transform of $B$ is $(2\pi)^{-d/2}(b\convolution b)$. The
above integral is $(2\pi)^{d/2}$ times the Fourier transform of
$x_i^2\cdot B$, evaluated at $0$.
Since multiplying a function by $i\cdot x_j$ 
corresponds to partial differentiation by $x_j$ in the Fourier domain,
$$ S = \sum_{i=1}^d \left(\frac{\partial^2}{\partial x_i^2}
(b\convolution b)\right)(0) = \sum_{i=1}^d
\left(\left(\frac{\partial}{\partial x_i} b\right)
\convolution \left(\frac{\partial}{\partial x_i} b\right)
\right)(0) = \sum_{i=1}^d \left\| \frac{\partial}{\partial x_i} b
\right\|_2^2$$
with the last equality using that $\frac{\partial}{\partial x_i} b$ is
odd.

We have, for $x$ in the unit ball,
$$ \left(\frac{\partial}{\partial x_i} b\right)(x) =
-2x_i $$
so that, after switching to hyperspherical coordinates,
\begin{equation}\EquationName{switch-hypersphere}
\sum_{i=1}^d \left\| \frac{\partial}{\partial x_i} b
\right\|_2^2 = A_d \cdot \int_{0}^1
4r^{d+1}dr .
\end{equation}
\begin{claim}\ClaimName{smallmoment}
$$\sum_{i=1}^d \left\| \frac{\partial}{\partial x_i} b
\right\|_2^2 = O(d^2) $$
\end{claim}
\begin{proof}
By definition of $b$,
\begin{eqnarray*}
\|b\|_2^2 &=& A_d \cdot \int_{0}^1
r^{d-1} +  r^{d+3} - 2r^{d+1}dr\\
&=& A_d\cdot \frac{8}{d(d+2)(d+4)}\\
&=& A_d\cdot \Omega(1/d^3).
\end{eqnarray*}
We also have by \Equation{switch-hypersphere} that
$$
\sum_{i=1}^d \left\| \frac{\partial}{\partial x_i} b
\right\|_2^2 = A_d\cdot \frac{4}{d+2} = A_d\cdot O(1/d) .$$
The claim follows since $\|b\|_2^2 = 1$.
\end{proof}
We now finish the proof of the lemma. Since $B$ has unit
integral on $\R^d$ (\Lemma{unit-integral}) and is nonnegative
everywhere, we can view $B$ as the density function of a probability
distribution on $\R^d$.  Then $S$ can be viewed as
$\E_{x\sim B}[\|x\|_2^2]$. Then by Markov's inequality, for $x\sim B$,
$$ \Pr\left[\|x\|_2^2 \ge z^2\cdot \E[\|x\|_2^2]\right] \le 1/z^2 ,$$
which is equivalent to
$$ \Pr\left[\|x\|_2 \ge z\cdot \sqrt{\E[\|x\|_2^2]}\right] \le 1/z^2 .$$
We conclude by observing that the above probability is simply
$$ \int_{\|x\|_2 \ge z \cdot \sqrt{\E[\|x\|_2^2]}} B(x)dx ,$$
from which the lemma follows since $\E[\|x\|_2^2] = O(d^2)$ by
\Claim{smallmoment}.
\end{proof}

We now state the main theorem of this section, which 
says that if $F$ is bounded, then
$\tilde{F}^c$ is smooth with strong bounds on its mixed partial
derivatives, and is
close to $F$ on points where $F$ satisfies
some continuity property.

\begin{theorem}\TheoremName{approx-uniform}
Let $F:\R^d\rightarrow\R$ be bounded and $c>0$ be arbitrary. Then,
\begin{itemize}
\item[i.] $\|\partial^\beta \tilde{F}^c\|_\infty \le \|F\|_\infty
  \cdot
  (2c)^{|\beta|}$ for all $\beta\in\N^d$.
\item[ii.] Fix some $x\in\R^d$.  Then if $|F(x)-F(y)|\le \eps$
  whenever $\|x-y\|_2 \le \delta$ for some $\eps,\delta\ge 0$, then
  $|\tilde{F}^c(x) - F(x)| \le
  \eps + \|F\|_\infty \cdot O(d^2 / (c^2\delta^2))$.
\end{itemize}
\end{theorem}
\begin{proof}
We first prove (i).
\begin{eqnarray}
\nonumber \left|\left(\partial^\beta \tilde{F}^c\right)(x)\right| &=&
\left|\left(\partial^\beta (B_c\convolution
  F)\right)(x)\right|\\
\nonumber &=&
\left|\left(\left(\partial^\beta B_c\right)\convolution F\right)(x)\right|\\
\nonumber &=& \left|\int_{\R^d}\left(\partial^\beta B_c\right)(y) F(x-y)dy\right|\\
\nonumber &\le& \|F\|_\infty \cdot \left\|\partial^\beta B_c\right\|_1\\
 &=& \|F\|_\infty \cdot c^{|\beta|}\cdot \left\|\partial^\beta
  B\right\|_1 \EquationName{bl1-isit}\\
\nonumber &\le& \|F\|_\infty\cdot (2c)^{|\beta|}
\end{eqnarray}
with the last inequality holding by \Lemma{l1bound}.

We now prove (ii).
\begin{eqnarray}
\nonumber \tilde{F}^c(x) &=& (B_c\convolution F)(x)\\
\nonumber &=& \int_{\R^d} B_c(x-y)F(y)dy \\
&=& F(x) + \int_{\R^d} (F(y) - F(x))B_c(x-y)dy\EquationName{use-unitint}\\
\nonumber &=& F(x) + \int_{\|x-y\|_2 < \delta} (F(y) - F(x))B_c(x-y) +
\int_{\|x-y\|_2 \ge \delta} (F(y) - F(x))B_c(x-y)\\
\nonumber &=& F(x) \pm \eps\cdot \int_{\|x-y\|_2 < \delta} |B_c(x-y)| +
\int_{\|x-y\|_2 \ge \delta} (F(y) - F(x))B_c(x-y)\\
\nonumber &=& F(x) \pm \eps\cdot \int_{\R^d} B_c(x-y) +
\int_{\|x-y\|_2 \ge \delta} (F(y) - F(x))B_c(x-y)\\
\nonumber &=& F(x) \pm \eps \pm \|F\|_\infty \cdot \int_{\|x-y\|_2 \ge \delta}
B_c(x-y)dy\\
\nonumber &=& F(x) \pm \eps \pm \|F\|_\infty \cdot \int_{\|u\|_2 \ge c\delta}
B(u)du\\
\nonumber &=& F(x) \pm \eps \pm \|F\|_\infty \cdot O(d^2/(c^2\delta^2))
\end{eqnarray}
where \Equation{use-unitint} uses \Lemma{unit-integral}.
\end{proof}

\begin{remark}\RemarkName{sharper-derivative-bounds}
It is possible to obtain sharper bounds on $\|\partial^\beta
\tilde{F}^c\|_\infty$. In particular, note in the proof of
\Theorem{approx-uniform} that $\|\partial^\beta
\tilde{F}^c\|_\infty \le \|F\|_\infty \cdot c^{|\beta|}\cdot
\|\partial^\beta B\|_1$.  An improved bound on
$\|\partial^\beta B\|_1$ versus that of \Lemma{l1bound}
turns out to be possible. This improvement is
useful when
FT-mollifying over high dimension, but in the proof of our main result
(\Theorem{main}) we are never concerned with $d>4$. We thus above
presented a simpler proof for clarity of exposition,
and we defer the details of the improvement
to \Section{improved-ftmol}.
\end{remark}

The following theorem is immediate from
\Theorem{approx-uniform}, and gives guarantees when FT-mollifying the
indicator function
of some region.
In \Theorem{approx-indicator}, and some later proofs
which invoke the theorem, we use the following notation.
For $R\subset \R^d$, we let $\partial R$ denote the boundary
of $R$ (specifically in this context, $\partial R$ is the set of
points $x\in\R^d$ such that for every $\eps>0$, the ball about $x$ of
radius $\eps$ intersects both $R$ and $\R^d\backslash R$). 

\begin{theorem}\TheoremName{approx-indicator}
For any region $R\subseteq \R^d$ and $x\in\R^d$,
$$ |I_R(x) - \tilde{I}_R^c(x)| \le \min\left\{1,
  O\left(\left(\frac{d}{c\cdot d_2(x,\partial
        R)}\right)^2\right)\right\} .$$
\end{theorem}
\begin{proof}
We have $|I_R(x) - \tilde{I}_R^c(x)| \le 1$ always. This follows since
$\tilde{I}_R^c$ is nonnegative (it is the convolution of nonnegative
functions), and is never larger than $\|I_R\|_\infty =
1$.
The other bound is obtained, for $x\notin\partial R$, by
applying \Theorem{approx-uniform} to $F = I_R$ with $\eps = 0$, $\delta =
d_2(x,\partial R)$.
\end{proof}

\section{A spectral moment bound for quadratic forms}\SectionName{spectral-moment}
For a quadratic form $p(x) = \sum_{i\le j}a_{i,j}x_ix_j$, we can associate a real symmetric matrix $A_p$ which has the
$a_{i,i}$ on the diagonals and $a_{\min\{i,j\},\max\{i,j\}}/2$ on the offdiagonals, so that $p(x) = x^TA_px$. We now show a
moment bound for quadratic forms which takes into account the maximum eigenvalue of $A_p$. Our proof is partly inspired by a
proof of Whittle \cite{Whittle60}, who showed the hypercontractive
inequality for degree-$2$ polynomials when comparing $q$-norms to
$2$-norms (see \Theorem{bonami-beckner}).

Recall the {\em Frobenius norm} of $A\in\R^{n\times n}$ is $\|A\|_2 = \sqrt{\sum_{i,j=1}^{n,n} A_{i,j}^2} = \sqrt{\sum_i \lambda_i^2} =
\sqrt{\tr(A^2)}$, where $\tr$ denotes trace and $A$ has eigenvalues $\lambda_1,\ldots,\lambda_n$. Also, let $\lambdamax{A}$ be the largest magnitude of an eigenvalue of $A$. We can now
state and prove the main theorem of this section, which plays a crucial role in our analysis of the regular case of our main theorem (\Theorem{main}).

\begin{theorem}\TheoremName{eigenbound}
Let $A\in\R^{n\times n}$ be symmetric
and $x\in\bits^n$ be random.  Then for all $k\geq 2$,
$$
\E[|(x^TAx) - \tr(A)|^k] \leq C^k \cdot \max\{\sqrt{k}\|A\|_2,
  k\lambdamax{A}\}^k
$$ where $C$ is an absolute constant. 
\end{theorem}

Note if $\sum_{i\le j}a_{i,j}^2 \le 1$ then $\lambdamax{A_p} \le 1$,
in which case our
bound recovers a similar moment bound as the one obtained
via hypercontractivity. Thus, in the special case of bounding $k$th
moments of degree-$2$ polynomials against their $2$nd moment, our
bound can be viewed as a
generalization of the hypercontractive inequality (and of Whittle's
inequality).

We first give two lemmas. The first is implied by Khintchine's
inequality \cite{Haagerup82}, and the second is a discrete analog of
one of Whittle's lemmas.

\begin{lemma}\LemmaName{deg1mom}
For $a\in\R^n$, $x$ as above, and $k\geq 2$ an even integer,
$ \E[(a^Tx)^k] \leq \|a\|_2^k \cdot k^{k/2}.$
\end{lemma}

\begin{lemma}\LemmaName{difflem}
If $X,Y$ are independent with $\E[Y]=0$ and if
$k\geq 2$, then
$\E[|X|^k] \leq \E[|X-Y|^k]$.
\end{lemma}
\begin{proof}
Consider the function $f(y) = |X-y|^k$.  Since $f^{(2)}$, the second derivative of $f$, is nonnegative
on $\R$, the claim  follows by Taylor's theorem since
$|X-Y|^k \geq |X|^k-k Y (\sgn(X)\cdot X)^{k-1}$.
\end{proof}

We are now prepared to prove our \Theorem{eigenbound}.

\begin{proofof}{\Theorem{eigenbound}}
Without loss of generality we can assume $\tr(A) = 0$.  This is
because if one considers $A' = A - (\tr(A)/n)\cdot I$, then
$x^TAx - \tr(A) = x^TA'x$, and we have $\|A'\|_2 \le \|A\|_2$ and
$\lambdamax{A'} \le 2\lambdamax{A}$.
We now start by proving our theorem for $k$ a power of 2 by induction on
$k$.  For $k=2$, $\E[(x^TAx)^2] =
4\sum_{i<j}A_{i,j}^2$ and $\|A\|_2^2 = \sum_i A_{i,i}^2 +
2\sum_{i<j}A_{i,j}^2$.  Thus $\E[(x^TAx)^2] \le 2\|A\|_2^2$.
Next we assume the statement of our
Theorem for $k/2$ and attempt to prove it for $k$.

We note that by \Lemma{difflem},
$$
\E[|x^TAx|^k] \leq \E[|x^TAx-y^TAy|^k] = \E[|(x+y)^TA(x-y)|^k],
$$
where $y\in\bits^n$ is random and independent of $x$.
Notice that if we swap $x_i$ with $y_i$ then $x+y$ remains constant as
does $|x_j-y_j|$ and that $x_i-y_i$ is replaced by its negation.
Consider averaging over all such swaps.  Let $\xi_i = ((x+y)^TA)_i$ and
$\eta_i = x_i - y_i$.  Let $z_i$ be
$1$ if we did not swap and $-1$ if we did.  Then
 $(x+y)^TA(x-y) = \sum_i \xi_i\eta_iz_i$.  Averaging over all swaps,
$$
\E_z[|(x+y)^TA(x-y)|^k] \le \left(\sum_i
  \xi_i^2\eta_i^2\right)^{k/2}\cdot k^{k/2}\le 2^kk^{k/2}\cdot \left(\sum_i
  \xi_i^2\right)^{k/2}.
$$
The first inequality is by \Lemma{deg1mom}, and the second uses
that $|\eta_i| \leq 2$.
Note that
$$
\sum_i \xi_i^2 = \|A(x+y)\|_2^2 \leq 2\|Ax\|_2^2+ 2\|Ay\|_2^2,
$$
and hence
$$\E[|x^TAx|^k] \leq 2^k\sqrt{k}^k \E[(2\|Ax\|_2^2+2\|Ay\|_2^2)^{k/2}]
\leq
4^k\sqrt{k}^k\E[(\|Ax\|_2^2)^{k/2}] ,$$
with the final inequality using Minkowski's inequality (namely that
$|\E[|X+Y|^p]|^{1/p} \le |\E[|X|^p]|^{1/p} + |\E[|Y|^p]|^{1/p}$ for
any random variables $X,Y$ and any $1\le p < \infty$).

Next note $\|Ax\|_2^2 = \langle Ax, Ax \rangle = x^TA^2 x$.  Let $B=A^2 -
\frac{\tr(A^2)}{n}I$.  Then $\tr(B) = 0$.  Also, $\|B\|_2
\le \|A\|_2\lambdamax{A}$ and
$\lambdamax{B} \leq
\lambdamax{A}^2$. The former holds since
$$\|B\|_2^2 = \left(\sum_i \lambda_i^4\right)
- \left(\sum_i \lambda_i^2\right)^2\Big/n \le \sum_i\lambda_i^4 \le
\|A\|_2^2\lambdamax{A}^2 .$$
The latter holds since the eigenvalues of $B$ are $\lambda_i^2 -
(\sum_{j=1}^n \lambda_j^2)/n$ for each $i\in[n]$.  The largest
eigenvalue of $B$ is thus at most that of $A^2$, and since $\lambda_i^2
\ge 0$, the smallest eigenvalue of $B$ cannot be smaller than
$-\lambdamax{A}^2$.

  We then have
$$
\E[(\|Ax\|_2^2)^{k/2}] = \E\left[\left| \|A\|_2^2 +
    x^TBx\right|^{k/2}\right] \leq
2^k\max\{\|A\|_2^k,\E[|x^TBx|^{k/2}]\}
.$$

Hence employing the inductive hypothesis on $B$ we have that
\begin{eqnarray*}
\E[|x^TAx|^k] &\leq& 8^k
\max\{\sqrt{k}\|A\|_2,C^{k/2}k^{3/4}\|B\|_2,C^{k/2}k\sqrt{\lambdamax{B}}\}^k\\
&\leq& 8^k C^{k/2}
\max\{\sqrt{k}\|A\|_2,k^{3/4}\sqrt{\|A\|_2\lambdamax{A}},k\lambdamax{A}\}^k\\
&=& 8^k C^{k/2}
\max\{\sqrt{k}\|A\|_2,k\lambdamax{A}\}^k ,
\end{eqnarray*}
with the final equality holding since the middle term above is the
geometric mean of the other two, and thus is dominated by at least one
of them.
This proves our hypothesis as long as $C\geq 64$.

To prove our statement for general $k$, set $k' = 2^{\ceil{\log_2 k}}$.
Then by the
power mean inequality and our results for $k'$ a power of $2$,
$ \E[|x^TAx|^k] \leq
(\E[|x^TAx|^{k'}])^{k/k'} \le 
128^k\max\{\sqrt{k}\|A\|_2,k\lambdamax{A}\}^k$.
\end{proofof}

\section{Fooling regular degree-$2$ threshold
  functions}\SectionName{regular}

The main theorem of this section is the following.

\begin{theorem}\TheoremName{main-regular}
Let $0<\eps<1$ be given. Let $X_1,\ldots,X_n$ be independent Bernoulli
and $Y_1,\ldots,Y_n$ be $2k$-wise independent Bernoulli for $k$ a
sufficiently large multiple of $1/\eps^{8}$.
If $p$ is multilinear and of degree $2$ with $\sum_{|S|>0}
\widehat{p}_S^2 = 1$, and
$\Inf_i(p) \le \tau$ for all $i$, then
$$ \E[\sgn(p(X))] - \E[\sgn(p(Y))] = O(\eps
+ \tau^{1/9}) .$$
\end{theorem}

Throughout this section, $p$ always refers to the polynomial of
\Theorem{main-regular}, and $\tau$ refers to the maximum influence of
any variable in $p$. Observe $p$ (over the hypercube) can be
written as $q+p_4+C$, where $q$ is a multilinear quadratic form, $p_4$
is a linear form, and $C$ is a constant. Furthermore,
$\|A_q\|_2 \le 1/2$ and $\sum_S \widehat{p_4}_S^2 \le 1$.
Using the spectral theorem for real symmetric matrices,
 we
write $p = p_1 - p_2 +
p_3 + p_4 + C$ where $p_1,p_2,p_3$ are quadratic forms satisfying
$\lambdamin(A_{p_1}),\lambdamin(A_{p_2}) \ge \delta$,
$\lambdamax{A_{p_3}} < \delta$, and $\|A_{p_i}\|_2 \le 1/2$ for $1\le i\le
3$, and also with $p_1,p_2$ positive semidefinite  (see
\Lemma{decompose-quadratic} for details on how this is
accomplished). Here $\lambdamin(A)$ denotes the smallest magnitude of
a non-zero eigenvalue of $A$.
Throughout this section we let
$p_1,\ldots,p_4$, $C$, $\delta$ be as discussed here.
We use $\eigensum$ to denote $\tr(A_{p_3})$.
The value $\delta$ will be set later in the proof of
\Theorem{main-regular}.

Throughout this section it will be notationally convenient to
define the map
$M_p:\R^n\rightarrow \R^4$ by $M_p(x) = (\sqrt{p_1(x)}, \sqrt{p_2(x)},
p_3(x) - \eigensum, p_4(x))$. Note the the first two coordinates of
$M_p(x)$ are indeed always real since $p_1,p_2$ are positive
semidefinite.

Before giving the proof of \Theorem{main-regular}, we first prove
\Lemma{kwise-fools}, which states that for $F:\R^4\rightarrow\R$,
$F(M_p(x))$ is fooled by
bounded independence as long as $F$ is even in $x_1,x_2$ and certain
technical conditions are satisfied. The proof of \Lemma{kwise-fools}
invokes the following lemma, which follows from lemmas in the Appendix
(specifically, by combining \Lemma{small-constant} and
\Lemma{boundmoment}).

\begin{lemma}\LemmaName{boundmoment2}
For a quadratic form $f$ and random $x\in\bits^n$,
$$ \E[|f(x)|^k] \le 2^{O(k)}\cdot (\|A_f\|_2k^k +
(\|A_f\|_2^2/\lambdamin(A_f))^{k}) .$$
\end{lemma}

\begin{lemma}\LemmaName{kwise-fools}
Let $\eps>0$ be arbitrary.
Let $F:\R^4\rightarrow\R$ be even in each of its first
two arguments such that
$\|\partial^{\beta} \tilde{F}^c\|_{\infty} = O(\alpha^{|\beta|})$
for all multi-indices $\beta\in\N^4$ and some $\alpha>1$. Suppose
$1/\delta \ge B\alpha$ for a
sufficiently large constant $B$.
Let
$X_1,\ldots,X_n$ be independent Bernoulli, and
$Y_1,\ldots,Y_n$ be $k'$-independent Bernoulli for $k' = 2k$ with $k \ge
\max\{\log(1/\eps),B\alpha/\sqrt{\delta},B\alpha^2\}$ an even integer.
Write $X=(X_1,\ldots,X_n)$ and $Y=(Y_1,\ldots,Y_n)$. Then
$ |\E[F(M_p(X))] -
\E[F(M_p(Y))]| < \eps
$.
\end{lemma}
\begin{proof}
We Taylor-expand $F$ to obtain a polynomial $P_{k-1}$ containing
all monomials up to degree $k-1$. Since $F(x)$ is even in $x_1,x_2$, we
can assume $P_{k-1}$ is a polynomial in $x_1^2,x_2^2,x_3,x_4$.
Let $x\in\R^4$ be arbitrary. We apply Taylor's theorem to bound
$R(x) = |F(x) - P_{k-1}(x)|$. Define $x_*
= \max_i\{|x_i|\}$. Then
\begin{eqnarray}
\nonumber R(x) &\le& \alpha^k\cdot \sum_{|\beta|=k}
\frac{|x_1|^{\beta_1} \cdot |x_2|^{\beta_2} \cdot
  |x_3|^{\beta_3}\cdot |x_4|^{\beta_4}}{\beta_1!\cdot \beta_2!\cdot \beta_3!\cdot \beta_4!}\\
\nonumber &\le&  \alpha^kx_*^k\cdot \sum_{|\beta|=k}
\frac{1}{\beta_1!\cdot \beta_2!\cdot \beta_3!\cdot \beta_4!} \\
\nonumber &=&  \alpha^kx_*^k\cdot\frac {1}{k!}\cdot \sum_{|\beta|=k}
\binom{k}{\beta_1,\ldots,\beta_4} \\
&\le& \alpha^k4^k\cdot\frac{x_1^k+x_2^k+x_3^k+x_4^k}{k!} \EquationName{taylor-process},
\end{eqnarray}
with the absolute values unnecessary in the last inequality since
$k$ is even.
We now observe
\begin{align*}
&|\E[F(M_p(X))] - \E[F(M_p(Y))]|\\
&\hspace{1in}\le
\alpha^k2^{O(k)}\cdot\frac{\E[(p_1(X))^{k/2}]
  +\E[(p_2(X))^{k/2}]+\E[(p_3(X)-\eigensum)^k] +
  \E[(p_4(X))^k]}{k^k}
\end{align*}
since (a) every term in
$P_{k-1}(M_p(X))$
is a monomial of degree at most $2k-2$ in the $X_i$, by evenness of
$P_{k-1}$ in $x_1,x_2$, and is thus determined by $2k$-independence,
(b)
$\sqrt{p_1(X)},\sqrt{p_2(X)}$ are real by positive
semidefiniteness of $p_1,p_2$ (note that we are only given that the
high order partial derivatives are bounded by  $O(\alpha^k)$ on the
reals; we have no guarantees for complex arguments), and (c) the
moment expectations above are equal for $X$ and $Y$ since they are
determined by $2k$-independence.

We now bound the error term above.  We have
$$ \E[(p_1(X))^{k/2}] = 2^{O(k)}(k^{k/2} +
\delta^{-k/2})$$
by \Lemma{boundmoment2}, with the same bound holding for
$\E[(p_2(X))^{k/2}]$. We also have
$$\E[(p_3(X) - \eigensum)^k] \le 2^{O(k)}\cdot
\max\{\sqrt{k},(\delta k)\}^k
$$
by \Theorem{eigenbound}. We finally have 
$$\E[(p_4(X))^k] \le k^{k/2}$$
by \Lemma{deg1mom}.  
Thus in total,
$$ |\E[F(M_p(X))] - \E[F(M_p(Y))]| \le 2^{O(k)}\cdot
((\alpha/\sqrt{k})^{k} + (\alpha/(k\sqrt{\delta}))^k +
(\alpha\delta)^k ) ,$$
which is at most $\eps$ for sufficiently large $B$ by our lower bounds
on $k$ and $1/\delta$.
\end{proof}

In proving \Theorem{main-regular}, we will need a lemma which states
that $p$ is anticoncentrated even when evaluated on
Bernoulli random variables which are $k$-wise independent.  To show this,
we make use of the following lemma, which
follows from the Invariance Principle, the
hypercontractive inequality, and the anticoncentration bound of
\cite{CW01}. The proof is in \Section{regular-proofs}.

\begin{lemma}\LemmaName{boundbad}
Let $\eta,\eta'\ge 0,t\in\R$ be given, and let $X_1,\ldots,X_n$ be independent
Bernoulli. Then
$$ \Pr[|p(X) - t| \le \eta\cdot(\sqrt{p_1(X)} + \sqrt{p_2(X)}
+ 1) + \eta'] = O(\sqrt{\eta'} + (\eta^2/\delta)^{1/4} + \tau^{1/9} +
\exp(-\Omega(1/\delta))) .$$
\end{lemma}

We now prove our anticoncentration lemma in the case of limited
independence.

\begin{lemma}\LemmaName{bounded-anticoncentration}
Let $\eps'$ be given.
Suppose $k \ge D/(\eps')^{4}$ for a sufficiently large
constant $D>0$.
Let $Y_1,\ldots,Y_n$ be $k$-wise independent Bernoulli, and let $t\in\R$ be
arbitrary.
Then
$$ \Pr[|p(Y)-t| < \eps'] \le O(\sqrt{\eps'} + \tau^{1/9}) .$$
\end{lemma}
\begin{proof}
Define the region $\smallball_{t,\eps'} = \{(x_1,x_2,x_3,x_4) : |x_1^2
- x_2^2 + x_3 + x_4 + C + \eigensum - t|
< \eps'\}$, and also the region $\largerball_{\rho,t,\eps'} = \{x :
d_2(x, \smallball_{t,\eps'}) \le \rho\}$ for $\rho\ge 0$.
Consider the
 FT-mollification $\tilde{I}^c_{\largerball_{\rho,t,\eps'}}$ of
 $I_{\largerball_{\rho,t,\eps'}}$ for $c = A/\rho$, with $A$ a 
 large constant to be determined later. We note a few properties of
 $\tilde{I}^c_{\largerball_{\rho,t,\eps'}}$:
\begin{itemize}
\item[i.] $\|\partial^\beta
  \tilde{I}^c_{\largerball_{\rho,t,\eps'}}\|_\infty \le (2c)^{|\beta|}$
\item[ii.] $ \tilde{I}^c_{\largerball_{\rho,t,\eps'}}(x) \ge \frac 12
  \cdot
  I_{\smallball_{t,\eps'}}(x)$
\item[iii.] $\tilde{I}^c_{\largerball_{\rho,t,\eps'}}(x) =
  \max\left\{1,O\left((c\cdot
    d_2(x,\smallball_{t,\eps'}))^{-2}\right)\right\}$ for any $x$ with
$d_2(x,\smallball_{t,\eps'}) \ge 2\rho$
\end{itemize}

Item (i) is straightforward from \Theorem{approx-uniform}. For item
(ii), note that if $x\in\smallball_{t,\eps'}$, then $d_2(x, \partial
\largerball_{\rho,t,\eps'}) \ge \rho$, implying
$$|\tilde{I}^c_{\largerball_{\rho,t,\eps'}}(x) - 1| =
O\left(\frac{1}{c^2\rho^2}\right), $$
which is at most $1/2$ for $A$ a sufficiently large
constant. Furthermore, $\tilde{I}^c_{\largerball_{\rho,t,\eps'}}$ is
nonnegative. Finally,
for (iii), by \Theorem{approx-indicator} we have
\begin{eqnarray*}
 \tilde{I}^c_{\largerball_{\rho,t,\eps'}}(x) &=&
  \max\left\{1,O\left((c\cdot
    d_2(x,\partial\largerball_{\rho,t,\eps'}))^{-2}\right)\right\} \\
&\le& \max\left\{1,O\left((c\cdot
    d_2(x,\largerball_{\rho,t,\eps'}))^{-2}\right)\right\} \\
&\le& \max\left\{1,O\left((c\cdot
    (d_2(x,\smallball_{t,\eps'})-\rho))^{-2}\right)\right\} \\
&\le& \max\left\{1,O\left((c\cdot
    d_2(x,\smallball_{t,\eps'}))^{-2}\right)\right\}
\end{eqnarray*}
with the last inequality using that $d_2(x,\smallball_{t,\eps'}) \ge
2\rho$.

Noting $\Pr[|p(Z)-t| < \eps'] = \E[I_{\smallball_{t,\eps'}}(M_p(Z))]$
 for any random variable $Z = (Z_1,\ldots,Z_n)$, item (ii) tells us that
\begin{equation}
 \Pr[|p(Z)-t|\le \eps']\le 2\cdot
\E[\tilde{I}^c_{\largerball_{\rho,t,\eps'}}(M_p(Z))] .\EquationName{crucial}
\end{equation}
We now proceed in two steps.  We first show
$\E[\tilde{I}^c_{\largerball_{\rho,t,\eps'}}(M_p(X))] =
O(\sqrt{\eps'} + \tau^{1/9})$ by applications of \Lemma{boundbad}. We
then show 
$\E[\tilde{I}^c_{\largerball_{\rho,t,\eps'}}(M_p(Y))] =
O(\sqrt{\eps'} + \tau^{1/9})$ by applying \Lemma{kwise-fools}, at
which point we will have proven our lemma via \Equation{crucial} with
$Z=Y$.

\bigskip

\noindent $\mathbf{\E[\tilde{I}^c_{\largerball_{\rho,t,\eps'}}(M_p(X))] =
O(\sqrt{\eps'} + \tau^{1/9})}$: We first observe that for $x\notin\smallball_{t,\eps'}$,
\begin{equation}
d_2(x,\smallball_{t,\eps'}) \ge \frac 12\cdot \min\left\{\frac{|x_1^2-x_2^2+x_3+x_4+C+\eigensum
  - t| - \eps'}{2(|x_1|+|x_2|+1)}, \sqrt{|x_1^2-x_2^2+x_3+x_4+C+\eigensum
  - t| - \eps'}\right\} .\EquationName{change-stuff}
\end{equation}
This is because by adding a vector $v$ to $x$, we
can
change each individual coordinate of $x$ by at most $\|v\|_2$, and can
thus change the value of $|x_1^2-x_2^2+x_3+x_4+C+\eigensum-t|-\eps'$
by at
most $2\|v\|_2\cdot (|x_1|+|x_2|+1) + \|v\|_2^2$.

Now let $X\in\bits^n$ be uniformly random.
We thus have that, for any particular $w>0$,
\begin{align*}
\Pr[0<d_2(M_p(X), \smallball_{t,\eps'}) \le
w] \le &\ \Pr\left[\min\left\{\frac{|p(X) -
    t| - \eps'}{2(\sqrt{p_1(X)}+\sqrt{p_2(X)}+1)},
  \sqrt{|p(X)-t|-\eps'}\right\} \le 2w\right]\\
\le &\ \Pr[|p(X) -
    t| \le
4w\cdot (\sqrt{p_1}(X)+\sqrt{p_2(X)}+1) + \eps']\\
&{} + \Pr[|p(X)-t| \le
4w^2 + \eps']\\
= &\ O(\sqrt{\eps'} + w + \sqrt{w} + (w^2/\delta)^{1/4} +
\tau^{1/9} + \exp(-\Omega(1/\delta)))
\end{align*}
with the last inequality holding by \Lemma{boundbad}.

Now, by item (iii), 
\begin{align}
\nonumber \E[\tilde{I}^c&_{\largerball_{\rho,t,\eps'}}(M_p(X))]\\
&{}\le
\nonumber \Pr[d_2(M_p(X), \smallball_{t,\eps'}) \le 2\rho] +
O\left(\sum_{s=1}^{\infty} 2^{-2s} \cdot \Pr[2^s\rho <
  d_2(M_p(X),\smallball_{t,\eps'}) \le 2^{s+1}\rho]\right)\\
\nonumber &{} \le O(\sqrt{\eps'} + \sqrt{\rho} + (\rho^2/\delta)^{1/4} +
\tau^{1/9} + \exp(-\Omega(1/\delta))\\
\nonumber &\hspace{.2in}{}+
O\Bigg(\sum_{s=1}^{\infty} 2^{-2s} \cdot (\sqrt{\eps'} + 2^{s+1}\rho +
  \sqrt{2^{s+1}\rho} + (2^{2s+2}\rho^2/\delta)^{1/4} + \tau^{1/9}+
  \exp(-\Omega(1/\delta)))\Bigg)\\
&{} = O(\sqrt{\eps'} + \sqrt{\rho} + (\rho^2/\delta)^{1/4} + \tau^{1/9}+
\exp(-\Omega(1/\delta)) \EquationName{thisissmall}
\end{align}
We now make the settings
$$\rho = (\eps')^2,\hspace{.3in} \frac{1}{\delta} = 2Bc =
\frac{2AB}{\rho} .$$
where $B>1$ is the
sufficiently large constant in \Lemma{kwise-fools}. Thus
\Equation{thisissmall} is now $O(\sqrt{\eps'} + \tau^{1/9})$. (We
remark that a different $\delta$ is
used when proving \Theorem{main-regular}.)

\bigskip

\noindent $\mathbf{\E[\tilde{I}^c_{\largerball_{\rho,t,\eps'}}(M_p(Y))] =
O(\sqrt{\eps'} + \tau^{1/9})}$: It suffices to show 
$$ \E[\tilde{I}^c_{\largerball_{\rho,t,\eps'}}(M_p(Y))]
\approx_{\eps} \E[\tilde{I}^c_{\largerball_{\rho,t,\eps'}}(M_p(X))]
.$$

We remark that
$\tilde{I}^c_{\largerball_{\rho,t,\eps'}}$ can be
assumed to be even in both $x_1,x_2$. If not, then consider the
symmetrization
\begin{equation}\EquationName{symmetrization}
(\tilde{I}^c_{\largerball_{\rho,t,\eps'}}(x_1,x_2,x_3,x_4) +
\tilde{I}^c_{\largerball_{\rho,t,\eps'}}(-x_1,x_2,x_3,x_4) +
\tilde{I}^c_{\largerball_{\rho,t,\eps'}}(x_1,-x_2,x_3,x_4) +
\tilde{I}^c_{\largerball_{\rho,t,\eps'}}(-x_1,-x_2,x_3,x_4))/4 ,
\end{equation}
which does not affect any of our properties (i),(ii), (iii).

Now, by our choice of $k,\delta$ and item (i), we have by
\Lemma{kwise-fools} (with
$\alpha=2c$) that
$$ |\E[\tilde{I}^c_{\largerball_{\rho,t,\eps'}}(M_p(X))] -
\E[\tilde{I}^c_{\largerball_{\rho,t,\eps'}}(M_p(Y))]|  < \eps' .$$
This completes our proof by applying \Equation{crucial} with $Z=Y$.
\end{proof}

The following Corollary is proven similarly as \Lemma{boundbad}, but
uses anticoncentration under bounded independence (which we just
proved in \Lemma{bounded-anticoncentration}).  The proof is
in \Section{regular-proofs}.

\begin{corollary}\CorollaryName{boundbad2}
Let $\eta,\eta'\ge 0$ be given, and let $Y_1,\ldots,Y_n$ be
$k$-independent Bernoulli for $k$ as in
\Lemma{bounded-anticoncentration} with $\eps' = \min\{\eta/\sqrt{\delta},
\eta'\}$. Also assume $k\ge \ceil{2/\delta}$. Then
$$ \Pr[|p(X) - t| \le \eta\cdot(\sqrt{p_1(X)} + \sqrt{p_2(X)}
+ 1) + \eta'] = O(\sqrt{\eta'} + (\eta^2/\delta)^{1/4} + \tau^{1/9} +
\exp(-\Omega(1/\delta))) .$$
\end{corollary}

We are now ready to prove the main theorem of this section.

\medskip

\begin{proofof}{\Theorem{main-regular}}
Consider the region $\reg\subset\R^4$ defined
by $\reg = \{(x_1,x_2,x_3,x_4) : x_1^2 - x_2^2 + x_3 + x_4 + C +
\eigensum \ge 0\}$. Then note that $I_{[0,\infty)}(p(x)) = 1$
if and only if $I_{\reg}(M_p(x)) = 1$. It thus
suffices to show that $I_{\reg}$ is fooled in expectation by bounded
independence.

We set
$\rho = \eps^4$, $c = 1/\rho$, and $1/\delta = 2Bc$ for $B$ the
constant in the statement of \Lemma{kwise-fools}.
 We now show a chain of
inequalities to give our theorem:
$$
  \E[I_{\reg}(M_p(X))] \approx_{\eps + \tau^{1/9}}
  \E[\tilde{I}^c_{\reg}(M_p(X))] \approx_{\eps}
  \E[\tilde{I}^c_{\reg}(M_p(Y))]  \approx_{\eps + \tau^{1/9}}
  \E[I_{\reg}(M_p(Y))] 
$$

\vspace{.1in}

\noindent $\mathbf{\E[I_{\reg}(M_p(X))] \approx_{\eps + \tau^{1/9}}
  \E[\tilde{I}_{\reg}^c(M_p(X))]}:$
Similarly to as
in the proof of \Lemma{bounded-anticoncentration},
$$d_2(x, \partial\reg) \ge \frac 12\cdot
\min\left\{\frac{|x_1^2-x_2^2+x_3+x_4+C+\eigensum|}{2(|x_1|+|x_2|+1)},
  \sqrt{|x_1^2-x_2^2+x_3+x_4+C+\eigensum|}\right\} ,$$
and thus by \Lemma{boundbad},
 \begin{eqnarray*}
\Pr[d_2(M_p(X),\partial\reg) \le w] &\le& \Pr[|p(X)| \le 4w
\cdot (\sqrt{p_1(X)} + \sqrt{p_2(X)} +
 1)] + \Pr[|p(X)| \le 4w^2] \\
&=& O(w + \sqrt{w} + (w^2/\delta)^{1/4}
 + \tau^{1/9} + \exp(-\Omega(1/\delta)))
\end{eqnarray*}
Now, noting $|\E[I_{\reg}(M_p(X))] -
\E[\tilde{I}_{\reg}^c(M_p(X))]|\le \E[|I_{\reg}(M_p(X))] -
\tilde{I}_{\reg}^c(M_p(X))|]$ and applying \Theorem{approx-indicator}, 
\begin{align*}
|\E[I_{\reg}&(M_p(X))] -
\E[\tilde{I}_{\reg}^c(M_p(X))]|\\
&{} \le \Pr[d_2(M_p(X),\partial\reg) \le
2\rho] +
O\left(\sum_{s=1}^{\infty} 2^{-2s} \cdot \Pr[2^s\rho <
  d_2(M_p(X),\partial\reg) \le 2^{s+1}\rho]\right)\\
&{} \le O(\sqrt{\rho} + (\rho^2/\delta)^{1/4} +
\tau^{1/9} + \exp(-\Omega(1/\delta))\\
&\hspace{.15in}{}+
O\left(\sum_{s=1}^{\infty} 2^{-2s} \cdot (\sqrt{2^{s+1}\rho} +
  (2^{2s+2}\rho^2/\delta)^{1/4} + \tau^{1/9}+
  \exp(-\Omega(1/\delta)))\right)\\
&{} = O(\eps + \tau^{1/9})
\end{align*}
by choice of $\rho,\delta$ and applications of \Lemma{boundbad}.

\vspace{.2in}

\noindent $\mathbf{\E[\tilde{I}_{\reg}^c(M_p(X))] \approx_{\eps}
  \E[\tilde{I}_\reg^c(M_p(Y))]}:$ As in \Equation{symmetrization},
we can assume $\tilde{I}_\reg^c$ is even in $x_1,x_2$.  We apply
\Lemma{kwise-fools} with $\alpha=2c$, noting that $1/\delta = B\alpha$ and
that our setting of $k$ is sufficiently large.

\vspace{.2in}

\noindent $\mathbf{\E[\tilde{I}_{\reg}^c(M_p(Y))] \approx_{\eps + \tau^{1/9}}
  \E[I_{\reg}(M_p(Y))]}:$ The argument is identical as with the first
inequality, except that we use \Corollary{boundbad2} instead of
\Lemma{boundbad}. We remark that we do have sufficient independence to
apply \Corollary{boundbad2} since, mimicking our analysis of the first
inequality, we have
\begin{align}
\nonumber \Pr[|p&(Y)| \le 4\rho
\cdot (\sqrt{p_1(Y)} + \sqrt{p_2(Y)} +
 1)] + \Pr[|p(Y)| \le 4\rho^2]\\
& \le \Pr[|p(Y)| \le 4\rho
\cdot (\sqrt{p_1(Y)} + \sqrt{p_2(Y)} +
 1)] + \Pr[|p(Y)| \le \eps^2]\EquationName{latter-sum}
\end{align}
since $\rho^2 = o(\eps^2)$ (we only changed the second summand).
To apply \Corollary{boundbad2} to
\Equation{latter-sum}, we need $k\ge \ceil{2/\delta}$, which is true,
and $k = \Omega(1/(\eps'')^4)$, for $\eps'' =
\min\{\rho/\sqrt{\delta}, \eps^2\} = \eps^2$, which is also
true. \Corollary{boundbad2} then tells us \Equation{latter-sum} is
$O(\eps + \tau^{1/9})$.
\end{proofof}

Our main theorem of this Section (\Theorem{main-regular}) also holds
under the case that the $X_i,Y_i$ are standard normal, and without
any error term depending on $\tau$. We give a proof
in \Section{gaussian-setting}, by reducing back to the Bernoulli case.

\section{Reduction to the regular case}\SectionName{main-thm}

In this section, we complete the proof of Theorem~\ref{thm:main}. We accomplish this by providing a reduction from the general
case to the regular case. In fact, such a reduction can be shown to hold for any degree $d\geq 1$ and establishes the
following:

\begin{theorem} \TheoremName{reg-reduction}
Suppose $K_d$-wise independence $\eps$-fools the class of
$\tau$-regular degree-$d$ PTF's, for some parameter $0<\tau \leq
\eps$. Then $(K_d+L_d)$-wise independence $\eps$-fools all degree-$d$
PTFs, where $L_d=(1/\tau) \cdot \big(d \log
(1/\tau)\big)^{O(d)}$.
\end{theorem}

Noting that $\tau$-regularity implies that the maximum influence of
any particular variable is at
most $d\cdot \tau$, \Theorem{main-regular} implies that degree-$2$
PTF's that are $\tau$-regular, for $\tau = O(\eps^{9})$, are
$\eps$-fooled by
$K_2$-wise independence for $K_2 = O(\eps^{-8}) = \poly(1/\eps)$. By plugging in $\tau = O(\eps^9)$ in the
above theorem we obtain \Theorem{main}. The proof of \Theorem{reg-reduction} is based on recent machinery from
\cite{DSTW09}\footnote{We note that \cite{MZ10} uses a similar approach to obtain their PRG's for degree-$d$ PTF's. Their
methods are not directly applicable in our setting, one reason being that that their notion of ``regularity'' is different from
ours.}. Here we give a sketch, with full details in \Section{reduction-proofs}.

\begin{proofof-sketch}{\Theorem{reg-reduction}} Any
boolean function $f$ on $\{-1,1\}^n$ can be expressed as a binary decision tree where each internal node is labeled by a
variable, every root-to-leaf path corresponds to a restriction $\rho$ that fixes the variables as they are set on the path, and
every leaf is labeled with the restricted subfunction $f_{\rho}$. The main claim is that, if $f$ is a degree-$d$ PTF, then it
has such a decision-tree representation with certain strong properties. In particular, given an arbitrary degree-$d$ PTF
$f=\sgn(p)$, by \cite{DSTW09} there exists a decision tree $\T$ of depth $(1/\tau) \cdot \big(d \log (1/\tau)\big)^{O(d)}$, so that with
probability $1-\tau$ over the choice of a random root-to-leaf path\footnote{A ``random root-to-leaf path'' corresponds to the
standard uniform random walk on the tree.} $\rho$, the restricted subfunction (leaf) $f_\rho=\sgn(p_\rho)$ is either a
$\tau$-regular degree-$d$ PTF or $\tau$-close to a constant function.

Our proof of \Theorem{reg-reduction} is based on the above structural lemma. Under the uniform distribution, there is some
particular distribution on the leaves (the tree is not of uniform height); then conditioned on the restricted variables the
variables still undetermined at the leaf are still uniform.  With $(K_d+L_d)$-wise independence, a random walk down the tree
arrives at each leaf with the same probability as in the uniform case (since the depth of the tree is at most $L_d$). Hence,
the probability mass of the ``bad'' leaves is at most $\tau\le\eps$ even under bounded independence. Furthermore, the induced distribution
on each leaf (over the unrestricted variables) is $K_d$-wise independent. Consider a good leaf. Either the leaf is
$\tau$-regular, in which case we can apply \Theorem{main-regular}, or it is $\tau$-close to a constant function. At this point
though we arrive at a technical issue. The statement and proof in \cite{DSTW09} concerning ``close-to-constant'' leaves holds
only under the uniform distribution. For our result, we need a stronger statement that holds under any distribution (on the
variables that do not appear in the path) that has sufficiently large independence. By simple modifications of the proof
in~\cite{DSTW09}, we show that the statement holds even under $O(d \cdot \log (1/\tau))$-wise independence.
\end{proofof-sketch}

\section{Fooling intersections of threshold functions}\SectionName{intersections}

Our approach also implies that the intersection of
halfspaces (or even degree-$2$ threshold
functions) is fooled by bounded independence.  While 
\Theorem{main-regular-gaussian}
implies that
$\Omega(\eps^{-8})$-wise independence fools GW
rounding, we can do much better by noting
that to fool GW rounding it suffices to fool the intersection
of two halfspaces under the Gaussian measure.

This is because in the GW rounding scheme for \maxcut, each
vertex $u$ is first mapped to a vector $x_u$ of unit norm, and the
side of a bipartition $u$ is placed in is decided by
$\sgn(\inprod{x_u,r})$ for a random Gaussian vector $r$. For a vertex
$u$, let $H_u^+$ be the halfspace $\inprod{x_u,r} > 0$, and let
$H_u^-$ be the halfspace $\inprod{-x_u,r} > 0$.  Then note that the
edge $(u,v)$ is cut if and only if $r\in(H_u^+\cap
H_v^-)\cup(H_u^-\cap H_v^+)$, i.e. $r$ must be in the union of the
intersection of two halfspaces.  Thus if we define the region
$R^+$ to be the topright quadrant of $\R^2$, and $R^-$ to be the
bottom left quadrant of $\R^2$, then we are interested in
fooling
$$\E[I_{R^+\cup R^-}(\inprod{x_u,r}, \inprod{-x_v,r})] =
\E[I_{R^+}(\inprod{x_u,r}, \inprod{-x_v,r})] +
\E[I_{R^+}(\inprod{-x_u,r}, \inprod{x_v,r})] ,$$
since the sum of such expectations over all edges $(u,v)$ gives us the
expected number of edges that are cut (note
equality holds above since the two halfspace intersections are
disjoint). The following theorem
then implies that to achieve a maximum cut within a factor $.878... -
\eps$ of
optimal in expectation, it suffices that the entries of the random normal
vector $r$ have entries that are $\Omega(1/\eps^2)$-wise independent.
The proof of
the theorem is in \Section{intersection-proofs}.

\begin{theorem}\TheoremName{intersection-halfspaces}
Let $H_1=\{x:\inprod{a,x}>\theta_1\}$ and
$H_2=\{x:\inprod{b,x}>\theta_2\}$ be two
halfspaces, with $\|a\|_2=\|b\|_2=1$.
Let $X,Y$ be $n$-dimensional vectors of standard normals with the
$X_i$ independent and the
$Y_i$ $k$-wise independent for $k = \Omega(1/\eps^2)$. Then $ |\Pr[X\in
H_1\cap H_2] -
\Pr[Y\in H_1\cap H_2]| < \eps.$
\end{theorem}

The proof of \Theorem{intersection-halfspaces}
can be summarized in one sentence: FT-mollify the indicator function
of $\{x : x_1\ge \theta_1,x_2\ge \theta_2\}\subset \R^2$.
We also in \Section{intersection-proofs} discuss how our proof of
\Theorem{intersection-halfspaces} easily generalizes to handle the
intersection of $m$
halfspaces, or even $m$ degree-$2$ PTF's, for any constant $m$, as
well as generalizations to case that $X,Y$ are Bernoulli vectors as
opposed to Gaussian. Our dependence on $m$ in all cases is polynomial.

\section*{Acknowledgments}
We thank Piotr Indyk and Rocco Servedio for comments that improved the
presentation of this work. We also thank Ryan O'Donnell for bringing
our attention to the problem of the intersection of threshold
functions.

\medskip

\bibliographystyle{plain}
\bibliography{allpapers}

\newpage
\appendix
\section{Basic linear algebra facts}\SectionName{lin-alg}

In this subsection we record some basic linear algebraic facts used in our proofs.

We start with two elementary facts.

\begin{fact}\FactName{preserve-eigen}
If $A,P\in\R^{n\times n}$ with $P$ invertible, then the eigenvalues of
$A$ and $P^{-1}AP$ are identical.
\end{fact}

\begin{fact}\FactName{trace-eigen}
For $A\in\R^{n\times n}$ with eigenvalues
$\lambda_1,\ldots,\lambda_n$, and for integer $k>0$, $\tr(A^k) =
\sum_i \lambda_i^k$.
\end{fact}

Note \Fact{preserve-eigen} and \Fact{trace-eigen} imply the following.

\begin{fact}\FactName{preserve-l2}
For a real matrix $A\in\R^{n\times n}$ and invertible matrix
$P\in\R^{n\times n}$,
$$ \|P^{-1}AP\|_2 = \|A\|_2 .$$
\end{fact}

The following standard result will be useful:

\begin{theorem}[Spectral Theorem {\cite[Section 6.4]{Strang09}}]\TheoremName{spectral}
If $A\in \R^{n\times n}$ is symmetric, there exists an orthogonal
$Q\in\R^{n\times n}$ with $\Lambda = Q^TAQ$ diagonal. In particular, all
eigenvalues of $A$ are real.
\end{theorem}

\begin{definition}
For a real symmetric matrix $A$, we define $\lambdamin(A)$ to be the
smallest magnitude of a non-zero eigenvalue of $A$ (in the case that
all eigenvalues are $0$, we set $\lambdamin(A) = 0$).  We define
$\lambdamax{A}$ to be the largest magnitude of an eigenvalue of $A$.
\end{definition}

We now give a simple lemma that gives an upper bound on the magnitude of the trace of a symmetric matrix with positive
eigenvalues.

\begin{lemma}\LemmaName{small-constant}
Let $A\in\R^{n\times n}$ be symmetric with
$\lambdamin(A)>0$. Then $|\tr(A)| \le \|A\|_2^2/\lambdamin(A)$.
\end{lemma}
\begin{proof}
We have
\begin{eqnarray*}
|\tr(A)| &=& \left|\sum_{i=1}^n \lambda_i\right|\\
&\le& \frac{\|A\|_2}{\lambdamin(A)}\cdot\sqrt{\sum_{i=1}^n\lambda_i^2}\\
&=& \frac{\|A\|_2^2}{\lambdamin(A)}
\end{eqnarray*}
We note $\sum_{i=1}^n \lambda_i^2 = \|A\|_2^2$,
implying the final equality. Also,
there are at most
$\|A\|_2^2/(\lambdamin(A))^2$ non-zero $\lambda_i$.
The sole inequality then follows by Cauchy-Schwarz.
\end{proof}

\section{Useful facts about polynomials}

\subsection{Facts about low-degree polynomials.}

We view $\bits^n$ as a probability space endowed with the uniform probability measure. For a function $f:\{-1,1\}^n\rightarrow
\R$ and $r \geq 1$, we let $\|f\|_r$ denote $(\E_x[|f(x)|^r])^{1/r}$.

Our first fact is a consequence of the well-known hypercontractivity theorem.

\begin{theorem}[Hypercontractivity \cite{Beckner75,Bonami70}]\TheoremName{bonami-beckner}
If $f$ is a degree-$d$ polynomial and $1\le r < q \le \infty$,
$$ \|f\|_q \le \sqrt{\frac{q-1}{r-1}}^d\|f\|_r .$$
\end{theorem}

Our second fact is an anticoncentration theorem for low-degree
polynomials over independent standard Gaussian random
variables. 

\begin{theorem}[Gaussian Anticoncentration \cite{CW01}]\TheoremName{anticoncentration}
For $f$ a non-zero, $n$-variate, degree-$d$ polynomial,
$$ \Pr[|f(G_1,\ldots,G_n)- t| \le \eps\cdot \Var[f]] = O(d\eps^{1/d}) $$
for all $\eps\in(0, 1)$ and $t\in\R$.  Here
$G_1,\ldots,G_n\sim\normal(0,1)$ are independent. (Here, and
henceforth, $\normal(\mu,\sigma^2)$ denotes the Gaussian distribution
with mean $\mu$ and variance $\sigma^2$.)
\end{theorem}

The following is a statement of the Invariance Principle of
Mossell, O'Donnell, and Oleszkiewicz \cite{MOO10}, in the special case
when the random variables $X_i$ are Bernoulli.

\begin{theorem}[Invariance Principle
  {\cite{MOO10}}]\TheoremName{invariance}
Let $X_1,\ldots,X_n$ be independent $\pm 1$ Bernoulli, and let $p$ be
a degree-$d$
multilinear polynomial with $\sum_{|S|>0} \widehat{p}_S^2 = 1$ and
$\max_i \Inf_i(p) \le \tau$.
Then
$$ \sup_t \left|\Pr[p(X_1,\ldots,X_n) \le t] - \Pr[p(G_1,\ldots,G_n)
  \le t]\right| = O(d\tau^{1/(4d+1)}) $$
where the $G_i\sim\normal(0,1)$ are independent.
\end{theorem}

The following tail bound argument is standard (see for example
\cite{AH09}). We repeat the argument here just to point out that only
bounded independence is required.

\begin{theorem}[Tail bound]\TheoremName{tailbound}
If $f$ is a degree-$d$ polynomial,
$t>8^{d/2}$, and $X$ is
drawn at random from a $(dt^{2/d})$-wise independent distribution over
$\{-1,1\}^n$, then
$$ \Pr[|f(X)| \ge t\|f\|_2] = \exp(-\Omega(dt^{2/d})) .$$
\end{theorem}
\begin{proof}
Suppose $k> 2$. By \Theorem{bonami-beckner},
$$ \E[|f(X)|^k] \le k^{dk/2}\cdot \|f\|_2^k ,$$
implying
\begin{equation}
\Pr[|f(X)| \ge t \|f\|_2] \le (k^{d/2}/t)^k \EquationName{boundpr}
\end{equation}
by Markov's inequality. Set $k =
2\cdot \floor{t^{2/d}/4}$ and note $k > 2$ as long as $t > 8^{d/2}$.
Now the right hand side of
\Equation{boundpr} is at most $2^{-dk/2}$, as desired. Finally, note
independence was only used to bound $\E[|f(X)|^k]$, which for $k$
even equals $\E[f(X)^k]$ and is thus determined by
$dk$-independence.
\end{proof}

\subsection{Facts about quadratic forms.}

The following facts are concerned with quadratic forms,
i.e. polynomials $p(x) = \sum_{i\le j}a_{i,j}x_ix_j$. We often
represent a quadratic form $p$ by its associated symmetric matrix
$A_p$, where
$$(A_p)_{i,j} = \begin{cases} a_{i,j}/2, \ & \ i< j
 \\
a_{j,i}/2, \ & \ i>j
 \\
 a_{i,j}, \ & \ i=j \end{cases} $$
so that $p(x) = x^TA_px$.

The following is a bound on moments for quadratic forms.

\begin{lemma}\LemmaName{boundmoment}
Let $f(x)$ be a degree-$2$ polynomial. Then, for $X=(X_1,\ldots,X_n)$
a vector of
independent Bernoullis,
$$ \E[|f(X)|^k] \le 2^{O(k)}(\|A_f\|_2k^k + |\tr(A_f)|^k) .$$
\end{lemma}
\begin{proof}
Over the hypercube we can write $f = q + \tr(A_f)$ where $q$ is
multilinear. Note $\|A_q\|_2
\le \|A_f\|_2$.  Then by \Theorem{bonami-beckner},
\begin{eqnarray*}
\E[|f(x)|^k] &=& \E[|q(x) + \tr(A_f)|^k]\\
&\le& \sum_{i=0}^k (\|A_f\|_2\cdot i)^i|\tr(A_f)|^{k-i}\\
&\le& \sum_{i=0}^k (\|A_f\|_2\cdot k)^i|\tr(A_f)|^{k-i}\\
&=& 2^{O(k)}\max\{\|A_f\|_2\cdot k, |\tr(A_f)|\}^k
\end{eqnarray*}
\end{proof}

The following corollary now follows from \Theorem{tailbound} and
\Lemma{small-constant}.

\begin{corollary}\CorollaryName{eigentail}
Let $f$ be a quadratic form with $A_f$ positive semidefinite,
$\|A_f\|_2 \le 1$, and $\lambdamin(A_f) \ge \delta$
for some $\delta\in(0, 1]$. Then, for $x$ chosen at random from a
$\ceil{2/\delta}$-independent family over $\{-1,1\}^n$,
$$ \Pr[f(x) > 2/\delta] = \exp(-\Omega(1/\delta)) .$$
\end{corollary}
\begin{proof}
Write $f = g + C$ via \Lemma{small-constant} with $0\le C\le
1/\delta$ and $g$ multilinear, $\|A_g\|_2\le \|A_f\|_2 \le 1$. Apply
\Theorem{tailbound} to $g$ with $t
= 1/\delta$.
\end{proof}

The following lemma gives a decomposition of any multi-linear quadratic form as a sum of quadratic forms with special
properties for the associated matrices. It is used in the proof of \Theorem{main-regular}.

\begin{lemma}\LemmaName{decompose-quadratic}
Let $\delta>0$ be given.
Let $f$ be a multilinear quadratic form.  Then $f$ can be written as
$f_1 - f_2 + f_3$ for
quadratic
forms $f_1,f_2,f_3$ where:
\begin{enumerate}
\item $A_{f_1},A_{f_2}$ are positive semidefinite with
  $\lambdamin(A_{f_1}),\lambdamin(A_{f_2}) \ge \delta$.
\item $\lambdamax{A_{f_3}} < \delta$.
\item $\|A_{f_1}\|_2,\|A_{f_2}\|_2,\|A_{f_3}\|_2 \le \|A_f\|_2$.
\end{enumerate}
\end{lemma}
\begin{proof}
Since $A_f$ is real and symmetric, we can find an orthogonal matrix $Q$
such that
$\Lambda = Q^{T}A_fQ$ is diagonal. Each diagonal entry
of $\Lambda$ is either at least
$\delta$, at most $-\delta$, or in between.  We create a matrix
$P$ containing all entries of $\Lambda$ which are at least
$\delta$, with the others zeroed out.  We similarly create $N$ to have
all entries at most $-\delta$.  We place the remaining entries in
$R$. We then set $A_{f_1} = QPQ^{T}, A_{f_2} = QNQ^{T}, A_{f_3} =
QRQ^{T}$. Note $\|\Lambda\|_2^2 = \|A_f\|_2^2$ by \Fact{preserve-l2},
so since we remove
terms from $\Lambda$ form each $A_{f_i}$, their Frobenius norms can
only shrink.
 The eigenvalue bounds hold by
construction and
\Fact{preserve-eigen}.
\end{proof}

\section{Why the previous approaches failed} \SectionName{previous-fail}

In this section, we attempt to provide an explanation as to why the
approaches of \cite{DGJSV09} and \cite{KNW10} fail to fool
degree-$2$ PTFs.

\subsection{Why the approximation theory approach failed}

The analysis in \cite{DGJSV09} crucially exploits the strong concentration and anti-concentration properties of the gaussian
distribution. (Recall that in the linear regular case, the random variable $\langle w, x\rangle$ is approximately Gaussian.)
Now consider a regular degree-$2$ polynomial $p$ and the corresponding PTF $f = \sgn(p)$. Since $p$ is regular, it has still
has ``good'' concentration and anti-concentration properties -- though quantitatively inferior than those of the Gaussian.
Hence, one would hope to argue as follows: use the univariate polynomial $P$ (constructed using approximation theory), allowing
its degree to increase if necessary, and carry out the analysis of the error as in the linear case.

The reason this fails is because the (tight) concentration properties of $p$ -- as implied by hypercontractivity -- are not
sufficient for the analysis to bound the error of the approximation, even if we let the degree of the polynomial $P$ tend to
infinity. (Paradoxically, the error coming from the worst-case analysis becomes worse as the degree of $P$ increases.)

Without going into further details, we mention that an additional problem for univariate approximations to work is this: the
(tight) anti-concentration properties of $p$ -- obtained via the Invariance Principle and the anti-concentration bounds of
~\cite{CW01} -- are quantitatively weaker than what is required to bound the error, even in the region where $P$ has small
point-wise error (from the $\sgn$ function).

\subsection{Why the analysis for univariate FT-mollification failed}
We discuss why the argument in \cite{KNW10} failed to generalize to
higher degree.  Recall that the argument was via the following chain
of inequalities:
\begin{equation}\EquationName{uber-chain}
\E[I_{[0,\infty)}(p(X))] \approx_{\eps}
\E[\tilde{I}^c_{[0,\infty)}(p(X))] \approx_{\eps}
\E[\tilde{I}^c_{[0,\infty)}(p(Y))] \approx_{\eps}
\E[I_{[0,\infty)}(p(Y))]
\end{equation}
The step that fails for high-degree PTFs is the second inequality
in \Equation{uber-chain}, which was argued by Taylor's theorem. Our
bounds on derivatives of
$\tilde{I}_{[0,\infty)}^c$, the FT-mollification of $I_{[0, \infty)}$
for a certain parameter $c = c(\eps)$ to make sure
$|I_{[0,\infty)} - \tilde{I}_{[0,\infty)}^c|<\eps$ ``almost
everywhere'', are such that
$||(\tilde{I}_{[0,\infty)}^c)^{(k)}||_{\infty} \ge 1$ for all $k$.
Thus, we have that the error term from Taylor's theorem is
at least $\E[(p(X))^k]/k!$.  The problem comes from the numerator.
Since we can assume the sum of squared coefficients of $p$ is
$1$ (note the $\sgn$ function is invariant to scaling of its
argument), known (and tight) moment bounds (via hypercontractivity)
only give us an upper
bound on $\E[(p(x))^k]$ which is larger than $k^{d
  k/2}$, where $\mathrm{degree}(p) = d$.  Thus, the error from
Taylor's theorem does not
 decrease to zero by increasing $k$ for $d\ge 2$,
since we only are able to divide by $k! \le
k^k$ (in fact, strangely, increasing the amount of independence $k$
{\em worsens} this bound).

\section{Proofs omitted from \Section{regular}}\SectionName{regular-proofs}

\subsection{Boolean setting.}

We next give a proof of \Lemma{boundbad}, where $p_1,p_2,\delta$ are as
in \Section{regular} (recall $p=p_1-p_2+p_3+p_4+C$ where $p_1,p_2$ are
positive semidefinite with minimum non-zero eigenvalues at least
$\delta$).

\noindent \BoldLemma{boundbad}
{\it
Let $\eta,\eta'\ge 0,t\in\R$ be given, and let $X_1,\ldots,X_n$ be independent
Bernoulli. Then
$$ \Pr[|p(X) - t| \le \eta\cdot(\sqrt{p_1(X)} + \sqrt{p_2(X)}
+ 1) + \eta'] = O(\sqrt{\eta'} + (\eta^2/\delta)^{1/4} + \tau^{1/9} +
\exp(-\Omega(1/\delta))) .$$
}
\begin{proof}
Applying \Corollary{eigentail}, we have
$$\Pr[\sqrt{p_1(X)} \ge
\sqrt{2/\delta}] = \exp(-\Omega(1/\delta)) ,$$
and similarly for $\sqrt{p_2(X)}$.
We can thus bound our desired probability by
$$ \Pr[|p(X) - t| \le 2\eta\sqrt{2/\delta} + \eta +
\eta'] + \exp(-\Omega(1/\delta)) .$$
By \Theorem{anticoncentration}, together with \Theorem{invariance}, we
can bound the probability in the lemma
statement by
$$O(\sqrt{\eta'} + (\eta^2/\delta)^{1/4} + \tau^{1/9} +
\exp(-\Omega(1/\delta))) . $$
\end{proof}

\noindent \BoldCorollary{boundbad2}
{\it
Let $\eta,\eta'\ge 0$ be given, and let $Y_1,\ldots,Y_n$ be
$k$-independent Bernoulli for $k$ as in
\Lemma{bounded-anticoncentration} with $\eps' = \min\{\eta/\sqrt{\delta},
\eta'\}$. Also assume $k\ge \ceil{2/\delta}$. Then
$$ \Pr[|p(X) - t| \le \eta\cdot(\sqrt{p_1(X)} + \sqrt{p_2(X)}
+ 1) + \eta'] = O(\sqrt{\eta'} + (\eta^2/\delta)^{1/4} + \tau^{1/9} +
\exp(-\Omega(1/\delta))) .$$
}
\begin{proof}
There were two steps in the proof of \Lemma{boundbad} which required
using the independence of the $X_i$. The first was in the
application of \Corollary{eigentail}, but that only required
$\ceil{2/\delta}$-wise independence, which is satisfied here. The next
was in using the anticoncentration of $p(X)$ (the fact that
$\Pr[|p(X)-t| < s] = O(\sqrt{s} + \tau^{1/9})$ for any $t\in\R$ and
$s>0$). However, given \Lemma{bounded-anticoncentration},
anticoncentration still holds under $k$-independence.
\end{proof}

\subsection{Gaussian Setting}\SectionName{gaussian-setting}
In the following Theorem we show that the conclusion of
\Theorem{main-regular} holds even under the Gaussian measure.

\begin{theorem}\TheoremName{main-regular-gaussian}
Let $0<\eps<1$ be given. Let $G=(G_1,\ldots,G_n)$ be a vector of independent
standard normal random variables,
and $G'=(G'_1,\ldots,G'_n)$ be a vector of $2k$-wise independent
standard normal random variables for $k$ a
sufficiently large multiple of $1/\eps^{8}$.
If $p(x) = \sum_{i\le j}a_{i,j}x_ix_j$ has $\sum_{i\le j}a_{i,j}^2 = 1$,
$$ \E[\sgn(p(G))] - \E[\sgn(p(G'))] = O(\eps) .$$
\end{theorem}
\begin{proof}
Our proof is by a reduction to the Bernoulli case, followed by an
application of \Theorem{main-regular}. We replace each $G_i$ with
$Z_i = \sum_{j=1}^N X_{i,j}/\sqrt{N}$ for a sufficiently large $N$ to be
determined later.  We also replace each $G_i'$ with $Z'_i =\sum_{j=1}^N
Y_{i,j}/\sqrt{N}$.  We determine these $X_{i,j},Y_{i,j}$ as
follows.  Let $\Phi:\R\rightarrow[0,1]$ be the cumulative distribution
function (CDF) of the standard normal. Define $T_{-1,N} = -\infty$,
$T_{N,N} = \infty$, and $T_{k,N} = \Phi^{-1}(2^{-N}\sum_{j=0}^k
\binom{N}{k})$ for $0\le k\le N$. Now, after a $G_i$ is chosen
according to a standard normal distribution, we identify the unique $k_i$ such
that $T_{k_i-1,N} \le G_i < T_{k_i,N}$. We then randomly select a
subset of $k_i$ of the $X_{i,j}$ to make $1$, and we set the others to
$-1$.  The $Y_{i,j}$ are defined similarly.  It should be noted that
the $X_{i,j},Y_{i,j}$ are Bernoulli random variables, with the
$X_{i,j}$ being independent and the $Y_{i,j}$ being $2k$-wise
independent. Furthermore, we define the $nN$-variate polynomial
$p':\{-1,1\}^{nN}\rightarrow \R$
to be the one obtained from this procedure, so that $p(G) = p'(X)$. We
then define $p''(x) = \alpha\cdot p'(x)$ for $\alpha = (\sum_{i<j}
a_{i,j}^2 + (1 - 1/N)\sum_i a_{i,i}^2)^{-1}$ so that the sum of
squared coefficients in $p''$ (ignoring constant terms, some of which
arise because the $x_{i,j}^2$ terms are $1$ on the hypercube) is
$1$. It should be observed that $1\le \alpha\le 1 + 1/(N-1)$.

Now, we make the setting $\epsilon = \log^{1/3}(N)/\sqrt{N}$.  By the
Chernoff bound,
\begin{equation}\EquationName{apply-chernoff}
\Pr[|k_i - N/2| \ge \epsilon N/2] = o(1) \hbox{ as } N
\hbox{ grows}.
\end{equation}
\begin{claim}\ClaimName{reallyclose}
If $(1-\epsilon)N/2\le k_i\le (1+\epsilon)N/2$, then
$|T_{k_i,N}-T_{k_i+1,N}| = o(1)$.
\end{claim}
Before proving the claim, we show how now we can use it to prove our
Theorem. We argue by the following chain of inequalities:
$$ \E[\sgn(p(G))] \approx_{\eps} \E[\sgn(p''(X))] \approx_{\eps}
\E[\sgn(p''(Y))] \approx_{\eps} \E[\sgn(p(G'))] .$$

\noindent $\mathbf{\E[\sgn(p(G))] \approx_{\eps} \E[\sgn(p''(X))]}:$
First we condition on the event $\mathcal{E}$ that $|Z_i - G_i| \le
\eps^3/n^2$ for all
$i\in[n]$; this happens with probability $1 - o(1)$ as $N$ grows by
coupling \Claim{reallyclose} and \Equation{apply-chernoff}, and
applying a union
bound over all $i\in[n]$. We
also condition on the event $\mathcal{E}'$ that $|G_i| =
O(\sqrt{\log(n/\eps)})$ for all
$i\in[n]$, which happens with probability $1 - \eps^2$ by a union
bound over $i\in[n]$ since a standard
normal random variable has probability $e^{-\Omega(x^2)}$ of being
larger than $x$ in absolute value. Now, conditioned on
$\mathcal{E},\mathcal{E'}$, we have
$$ |p(G) - p''(X)| \le n^2(\eps^3/n^2)^2 + (\eps^3/n^2)\sum_i
|G_i|\left(\sum_j |a_{i,j}|\right) \le \eps^2 + (\eps^3/n^2)\cdot
O(\sqrt{\log(n/\eps)})
\cdot \sum_{i,j}|a_{i,j}| .$$
We note $\sum_{i,j}a_{i,j}^2 = 1$, and thus $\sum_{i,j}|a_{i,j}| \le
n$ by Cauchy-Schwarz. We thus have that $|p'(X) - p(G)| \le \eps^2$ with
probability at least $1 - \eps^2$, and thus $|p''(X) - p(G)| \le
\eps^2 + |(\alpha-1) \cdot p(X)|$ with probability at least $1-\eps^2$.
We finally condition on the event $\mathcal{E}''$ that $|(\alpha-1)\cdot
p'(X)| \le\eps^2$. Since $p'$ can be written as a multilinear
quadratic form with sum of squared coefficients at most $1$, plus its
trace $\tr(A_{p'})$ (which is $\sum_i a_{i,i} \le \sqrt{n}$, by
Cauchy-Schwarz), we
have
$$\Pr[|(\alpha-1)\cdot p'(X)| \ge \eps^2] \le \Pr[|p'(X)| \ge
\eps^2\cdot (N-1)] = o(1),$$
which for large enough $N$ and the fact that $\|p'\|_2  = O(1 +
\tr(A_{p'}))$ irrespective of $N$, is at most
$$ \Pr[|p'(X)| \ge c\cdot \log(1/\eps)\|p'\|_2] ,$$
for a constant $c$ we can make arbitrarily large by increasing $N$.
We thus have $\Pr[\mathcal{E}''] \ge 1-\eps^2$ by \Theorem{tailbound}.
Now, conditioned on $\mathcal{E}\wedge\mathcal{E}'\wedge\mathcal{E}''$,
$\sgn(p''(X)) \neq \sgn(p(G))$ can only occur if $|p''(X)|
= O(\eps^2)$.
However, by anticoncentration (\Theorem{anticoncentration}) and the
Invariance Principle (\Theorem{invariance}), this
occurs with probability $O(\eps)$ for $N$ sufficiently large (note the
maximum influence of $p''$ goes to $0$ as $N\rightarrow\infty$).

\vspace{.2in}

\noindent $\mathbf{\E[\sgn(p''(X))] \approx_{\eps}
  \E[\sgn(p''(Y))]}:$
Since the maximum influence $\tau$ of any $x_{i,j}$ in
$p''$ approaches
$0$ as $N\rightarrow \infty$, we can apply \Theorem{main-regular} for
$N$ sufficiently large (and thus $\tau$ sufficiently small).

\vspace{.2in}

\noindent $\mathbf{\E[\sgn(p''(Y))] \approx_{\eps}
  \E[\sgn(p(G'))]}:$
This case is argued identically as in the first inequality, except
that we use anticoncentration of $p''(Y)$, which follows from
\Lemma{bounded-anticoncentration}, and we should ensure that we have
sufficient independence to
apply \Theorem{tailbound} with $t=O(\log(1/\eps))$, which we do.

\begin{proofof}{\Claim{reallyclose}}
The claim is argued by showing that for $k_i$ sufficiently close to
its expectation (which is $N/2$), the density function of the
Gaussian (i.e. the derivative of its CDF) is sufficiently large that
the distance we must move from $T_{k_i,N}$ to $T_{k_i+1, N}$ to change
the CDF by $\Theta(1/\sqrt{N})\ge 2^{-N}\binom{N}{k_i+1}$ is small. We argue
the case
$(1-\epsilon)N/2 \le k_i \le N/2$ since the case $N/2\le k_i\le
(1+\epsilon)N/2$ is argued symmetrically.  Also, we consider only the
case $k_i = (1-\epsilon)N/2$ exactly, since the magnitude of
the standard normal density function is smallest in this case.

Observe that each $Z_i$ is a degree-$1$ polynomial in the
$X_{i,j}$ with maximum influence $1/N$, and thus by the
Berry-Ess\'{e}en Theorem,
$$ \sup_{t\in\R} |\Pr[Z_i \le t] - \Pr[G_i \le t]| \le
\frac{1}{\sqrt{N}} .$$
Also note that
$$ \Pr[G_i \le T_{k_i,N}] = \Pr\left[Z_i \le \frac{2k_i}{\sqrt{N}} -
  \sqrt{N}\right]$$
by construction. We thus have
\begin{eqnarray*}
\Pr[G_i \le T_{k_i,N}] &=& \Pr\left[G_i \le \frac{2k_i}{\sqrt{N}} -
  \sqrt{N}\right] \pm \frac{1}{\sqrt{N}}\\
&=& \Pr[G_i \le \log^{1/3}(N)] \pm \frac{1}{\sqrt{N}}
\end{eqnarray*}
By a similar argument we also have
$$
\Pr[G_i \le T_{k_i+1,N}] =
\Pr\left[G_i \le \log^{1/3}(N) + \frac{2}{\sqrt{N}}\right] \pm
\frac{1}{\sqrt{N}}
$$
Note though for $t = \Theta(\log^{1/3}(N))$, the density function $f$ of
the standard normal satisfies $f(t) = e^{-t^2/2} =
N^{-o(1)}$. Thus, in this regime we can change the CDF by
$\Theta(1/\sqrt{N})$ by moving only $N^{o(1)}/\sqrt{N} = o(1)$ along
the real axis, implying $T_{k_i+1,N} - T_{k_i,N} = o(1)$.
\end{proofof}
\end{proof}

\section{Proofs from \Section{main-thm}}\SectionName{reduction-proofs}
\subsection{Proof of \Theorem{reg-reduction}}

We begin by stating the following structural lemma:

\begin{theorem} \TheoremName{regularity-kwise}
Let $f(x) = \sgn(p(x))$ be any degree-$d$ PTF.  Fix any $\tau
> 0.$ Then $f$ is equivalent to a decision tree $\T$ of
depth $\depth(d,\tau) \eqdef (1/\tau) \cdot(d \log(1/\tau))^{O(d)}$
with variables at the internal nodes and a degree-$d$ PTF
$f_\rho = \sgn(p_\rho)$ at each leaf $\rho$,  with the following
property: with probability at least $1 - \tau$, a random path
from the root reaches a leaf $\rho$ such that either: (i) $f_\rho$ is
$\tau$-regular degree-$d$ PTF, or (ii) For any $O(d \cdot
\log(1/\tau))$-independent distribution $\D'$ over $\bits^{n-|\rho|}$
there exists $b \in \bits$ such that $\Pr_{x \sim
\D'}[f_{\rho}(x) \neq b] \leq \tau$.
\end{theorem}

We now prove \Theorem{reg-reduction} assuming \Theorem{regularity-kwise}. We will need some notation. Consider a leaf of the
tree $\T$. We will denote by $\rho$ both the set of variables that appear on the corresponding root-to-leaf path and the
corresponding partial assignment; the distinction will be clear from context. Let $|\rho|$ be the number of variables on
the path. We identify a leaf $\rho$ with the corresponding restricted subfunction $f_{\rho} = \sgn(p_{\rho})$. We call a leaf
``good'' if it corresponds to either a $\tau$-regular PTF or to a ``close-to constant'' function. We call a leaf ``bad''
otherwise. We denote by $L(\T)$, $GL(\T)$, $BL(\T)$ the sets of leaves, good leaves and bad leaves of $\T$ respectively.

In the course of the proof we make repeated use of the following
standard fact:

\begin{fact}\FactName{projectingkwise}
Let $\D$ be a $k$-wise independent distribution over
$\bits^n$. Condition on any fixed values for any $t \leq k$ bits of
$\D$,
and let $\D'$ be the projection of $\D$ on the other $n-t$ bits. Then
$\D'$ is $(k-t)$-wise independent.
\end{fact}

Throughout the proof, $\D$ denotes a $(K_d+L_d)$-wise independent distribution over $\bits^n$. Consider a random walk on the
tree $\T$. Let $LD(\T, \D)$ (resp. $LD(\T, \U)$) be the leaf that the random walk will reach when the inputs are drawn from the
distribution $\D$ (resp. the uniform distribution). The following straightforward lemma quantifies the intuition that these
distributions are the same. This holds because the tree has small depth and $\D$ has sufficient independence.

\begin{lemma} \LemmaName{distr-leaves}
For any leaf $\rho \in L(\T)$ we have $\Pr \big[ LD(\T, \D) = \rho \big] =  \Pr \big[LD(\T, \U) = \rho \big].$
\end{lemma}

The following lemma says that, if $\rho$ is a good leaf, the
distribution induced by $\D$ on $\rho$ $O(\eps)$-fools the
restricted subfunction $f_{\rho}$.

\begin{lemma} \LemmaName{kwise-leaves}
Let $\rho \in GL(\T)$ be a good leaf and consider the projection $\D_{[n] \setminus \rho}$ of $\D$ on the variables not in
$\rho$. Then we have $\big| \Pr_{x \sim \D_{[n]\setminus \rho}} [f_{\rho}(x) =1]  -  \Pr_{y \sim \U_{[n]\setminus \rho}}[
f_{\rho}(y)=1] \big| \leq 2 \eps.$
\end{lemma}

\begin{proof}
If $f_{\rho}$ is $\tau$-regular, by \Fact{projectingkwise} and recalling that $|\rho| \leq \depth(d,\tau) \leq L_d$, the
distribution $\D_{[n]\setminus \rho}$ is $K_d$-wise independent. Hence, the statement follows by assumption. Otherwise,
$f_{\rho}$ is $\eps$-close to a constant, i.e. there exists $b \in \bits$ so that for any $t=O(d \log (1/\tau))$-wise
distribution $\D'$ over $\bits^{n-|\rho|}$ we have $\Pr_{x \sim \D'}[f_{\rho}(x) \neq b] \leq \tau$ $(*)$. Since $L_d >> t$,
\Fact{projectingkwise} implies that $(*)$ holds both under $\D_{[n]\setminus \rho}$ and $\U_{[n]\setminus \rho}$, hence the
statement follows in this case also, recalling that $\tau \leq \eps$.
\end{proof}

The proof of \Theorem{reg-reduction} now follows by a simple averaging argument. By the decision-tree decomposition of
\Theorem{regularity-kwise}, we can write
\[\Pr_{x\sim \D'_n} [f(x)=1]  = \sum_{\rho \in L(T)} \Pr \big[LD(\T, \D') = \rho \big] \cdot \Pr_{y \in
\D'_{[n]\setminus{\rho}}}\big[ f_{\rho}(y)=1 \big]\] where $\D'$ is either $\D$ or the uniform distribution $\U$. By
\Theorem{regularity-kwise} and \Lemma{distr-leaves} it follows that the probability mass of the bad leaves is at most $\eps$
under both distributions. Therefore, by \Lemma{distr-leaves} and \Lemma{kwise-leaves} we get
\begin{eqnarray*}
&&\Big| \Pr_{x\sim\D}[f(x)=1] - \Pr_{x\sim\U}[f(x)=1] \Big| \leq \eps
+ \\
&&\sum_{\rho \in GL(T)} \Pr \big[LD(\T, \U) = \rho \big] \cdot \big| \Pr_{y \in \U_{[n]\setminus{\rho}}}\big[ f_{\rho}(y)=1
\big] - \Pr_{y \in \D_{[n]\setminus{\rho}}}\big[ f_{\rho}(y)=1 \big] \big| \leq 3\eps.
\end{eqnarray*}
This completes the proof of \Theorem{reg-reduction}.

\subsection{Proof of \Theorem{regularity-kwise}} \SectionName{DSTW-mod}

In this section we provide the proof of \Theorem{regularity-kwise}. For the sake of completeness, we give below the relevant
machinery from \cite{DSTW09}. We note that over the hypercube every polynomial can be assumed to be multilinear, and so
whenever we discuss a polynomial in this section it should be assumed to be multilinear. We start by defining the notion of the
critical index of a polynomial:

\begin{definition}[critical index]\DefinitionName{tau-reg}
Let $p: \bits^n \to \R$ and $\tau>0$. Assume the variables are ordered such that $\Inf_i(p)\geq \Inf_{i+1}(p)$ for all $i \in
[n-1]$.  The {\em $\tau$-critical index} of $p$ is the least $i$ such that:
\begin{equation}
\frac{\Inf_{i+1}(p)}{\sum_{j=i+1}^n\Inf_j(p)} \leq \tau . \EquationName{reg}
\end{equation}
If \Equation{reg} does not hold for any $i$ we say that the
$\tau$-critical index of $p$ is $+ \infty.$ If $p$ is has
$\tau$-critical index 0, we say that $p$ is {\em $\tau$-regular}.
\end{definition}

We will be concerned with polynomials $p$ of degree-$d$. The work in
\cite{DSTW09} establishes useful random restriction lemmas for
low-degree polynomials. Roughly, they are as
follows: Let $p$ be a degree-$d$ polynomial. If the $\tau$-critical index of $p$ is zero, then $f = \sgn(p)$ is $\tau$-regular
and there is nothing to prove.
\begin{itemize}
\item If the $\tau$-critical index of $p$ is ``very large'', then a random
restriction of ``few'' variables causes $f = \sgn(p)$ to become a
``close-to-constant'' function with probability $1/2^{O(d)}$.
We stress that the distance between functions is measured in
\cite{DSTW09}
with respect to the uniform distribution on inputs. As
previously mentioned, we extend this statement to hold for any
distribution with sufficiently large independence.

\item If the $\tau$-critical index of $p$ is positive but
not ``very large'', then a random restriction of a ``small'' number of
variables -- the variables with largest influence in $p$
-- causes $p$ to become ``sufficiently'' regular with probability
$1/2^{O(d)}.$

\end{itemize}

Formally, we require the following lemma which is a strengthening of
Lemma 10 in \cite{DSTW09}:

\begin{lemma} \LemmaName{dichotomy-kwise}
Let $p:\bits^n \rightarrow \R$ be a degree-$d$ polynomial and assume
that its variables are in order of non-increasing
influence. Let $0< \tau', \beta <1/2$ be parameters. Fix $\alpha =
\Theta ( d \log \log (1/\beta) + d \log d)$ and $\tau'' =
\tau' \cdot (C' d \ln d \ln (1/\tau'))^d$, where $C'$ is a universal
constant. One of the following statements holds true:
\begin{enumerate}
\item The function $f = \sgn(p)$ is $\tau'$-regular.

\item  With probability at least $1/2^{O(d)}$
over a random restriction $\rho$ fixing the first $L' = \alpha/\tau'$
variables of $p$, the function $f_{\rho}=\sgn (p_{\rho})$
is $\beta$-close to a constant function. In particular, under any $O(d
\log (1/\beta))$-wise independent distribution $\D'$
there exists $b \in \bits$ such that $\Pr_{x \sim \D'}[f_{\rho}(x)
\neq b] \leq \tau'$.

\item There exists a value $k \leq \alpha/\tau'$,
such that with probability at least $1/2^{O(d)}$ over a random
restriction $\rho$ fixing the first $k$ variables of $p$, the
polynomial $p_{\rho}$ is $\tau''$-regular.
\end{enumerate}
\end{lemma}

By applying the above lemma in a recursive manner we obtain \Theorem{regularity-kwise}. This is done exactly as in the proof of
Theorem 1 in \cite{DSTW09}. We remark that in every recursive application of the lemma, the value of the parameter $\beta$ is
set to $\tau$. This explains why $O(d \log (1/\tau))$-independence suffices in the second statement of
\Theorem{regularity-kwise}. Hence, to complete the proof of \Theorem{regularity-kwise}, it suffices to establish
\Lemma{dichotomy-kwise}.

\begin{proofof}{\Lemma{dichotomy-kwise}}
We now sketch the proof of the lemma. The first statement of the lemma corresponds to the case that the value $\ell$ of
$\tau'$-critical index is $0$, the second to the case that it is $\ell
>L'$ and the third to $1\leq \ell\leq L'$.

The proof of the second statement proceeds in two steps. Let $H$ denote the first $L'$ most influential variables of $p$ and
$T=[n]\setminus H$. Let $p'(x_H) = \sum_{S \subseteq H} \widehat{p}(S) x_S $. We first argue that with probability at least
$2^{-\Omega(d)}$ over a random restriction $\rho$ to $H$, the restricted polynomial $p_{\rho}(x_T)$ will have a ``large''
constant term $\widehat{p}_{\rho}(\emptyset) = p'(\rho)$, in particular at least $\theta = 2^{-\Omega(d)}$.  The proof is based
on the fact that, since the critical index is large, almost all of the Fourier weight of the polynomial $p$ lies in $p'$, and
it makes use of a certain anti-concentration property over the hypercube. Since the randomness is over $H$ and the projection
of $\D$ on those variables is still uniform, the argument holds unchanged under $\D$.

In the second step, by an application of a concentration bound, we show that for at least half of these restrictions to $H$ the
surviving (non-constant) coefficients of $p_{\rho}$, i.e. the Fourier coefficients of the polynomial $p_{\rho}(x_T)-p'(\rho)$,
have small $\ell_2$ norm; in particular, we get that $\|p_{\rho} - p'_{\rho}\|_2 \leq \log(1/ \beta)^{-d}$. We call such
restrictions good. Since the projection of $\D$ on these ``head'' variables is uniform, the concentration bound applies as is.

Finally, we need to show that, for the good restrictions, the event the ``tail'' variables $x_T$ change the value of the
function $f_{\rho}$, i.e. $\sgn(p_{\rho}(x_T)+p'(\rho)) \neq \sgn (p'(\rho))$ has probability at most $\beta$. This event has
probability at most
\[ \Pr_{x_T} [|p_{\rho}(x_T) - p'(\rho)| \geq \theta] .\]
This is done in \cite{DSTW09} using a concentration bound on the ``tail'', assuming full independence. Thus, in this case, we
need to modify the argument since the projection of $\D$ on the ``tail'' variables is not uniform. However, a careful
inspection of the parameters reveals that the concentration bound needed above actually holds even under an assumption of $O(d
\log (1/\beta))$-independence for the ``tail'' $x_T$. In particular, given the upper bound on $\|p_{\rho} - p'_{\rho}\|_2$ and
the lower bound on $\theta$, it suffices to apply \Theorem{tailbound}
for $t = \log(1/\beta)^{d/2}$, which only requires
$(dt^{2/d})$-wise independence. Hence, we are done in this case too.

The proof of the third statement remains essentially unchanged for the
following reason: One proceeds by considering a random
restriction of the variables of $p$ up to the $\tau$-critical index --
which in this case is small. Hence, the distribution
induced by $\D$ on this space is still uniform. Since the randomness
is over these ``head'' variables, all the arguments remain
intact and the claim follows.
\end{proofof}
\section{Appendix to \Section{intersections}}\SectionName{intersection-proofs}

We show a generalization of \Theorem{intersection-halfspaces} to the
intersection of $m>1$ halfspaces, which implies
\Theorem{intersection-halfspaces} as the special case $m=2$.

\noindent \BoldTheorem{intersection-halfspaces}
{\it
Let $m>1$ be an integer.
Let $H_i=\{x:\inprod{a_i,x}>\theta_i\}$ for $i\in[m]$, with
$\|a_i\|_2=1$ for all $i$.
Let $X$ be a vector of $n$
i.i.d. Gaussians, and $Y$ be a vector of $k$-wise independent
Gaussians. Then for $k = \Omega(m^6/\eps^2)$,
$$ |\Pr[X\in \cap_{i=1}^mH_i] - \Pr[Y\in \cap_{i=1}^m H_i]| < \eps$$
}
\begin{proof}
Let $F:\R^n\rightarrow\R^m$ be the map $F(x) =
(\inprod{a_1,x},\ldots,\inprod{a_m,x})$, and let $R$ be the region $\{x :
\forall i\ x_i>\theta_i\}$.
Similarly as in the proof of \Theorem{main-regular}, we simply
show a chain of inequalities after setting $\rho = \eps/m$ and $c =
m/\rho$:
\begin{equation}\EquationName{chain2}
\E[I_{R}(F(X))] \approx_{\eps} \E[\tilde{I}_{R}^c(F(X))]
\approx_{\eps}  \E[\tilde{I}_{R}^c(F(Y))]  \approx_{\eps}
\E[I_{R}(F(Y))] .
\end{equation}
Note the maximum influence $\tau$ does not play a role since
under the Gaussian
measure we never need invoke the Invariance Principle. 
For the
first inequality, observe $d_2(x,\partial R) \ge
\min_i\{|x_i-\theta_i|\}$.
Then by a union bound,
$$ \Pr[d_2(F(X),\partial R) \le w] \le \Pr[\min_i\{|\inprod{a_i,X} -
\theta_i|\} \le w] \le
\sum_{i=1}^m\Pr[|\inprod{a_i,X} - \theta_i| \le w] ,$$
which is $O(m w)$ by \Theorem{anticoncentration} with $d=1$.  Now,

\begin{eqnarray}
\nonumber |\E[I_R(F(X))] - \E[\tilde{I}_R^c(F(X))]| &\le& \E[|I_R(F(X))] -
\tilde{I}_R^c(F(X))|]\\
\nonumber &\le& 
\Pr[d_2(F(X),\partial R) \le 2\rho]\\
&&\hspace{.2in}{} + 
O\left(\sum_{s=1}^\infty \left(\frac{m^2}{c^22^{2s}\rho^2}\right) \cdot
  \Pr[d_2(F(X),\partial R) \le
  2^{s+1}\rho]\right)\EquationName{use-approx-indicator}\\
\nonumber&=& 
\Pr[d_2(F(X),\partial R) \le 2\rho] + 
O\left(\sum_{s=1}^\infty 2^{-2s} \cdot \Pr[d_2(F(X),\partial R) \le
  2^{s+1}\rho]\right)\\
\nonumber&=& O(m\rho)\\
\nonumber&=& O(\eps)
\end{eqnarray}
where \Equation{use-approx-indicator} follows from
\Theorem{approx-indicator}.

The last inequality in \Equation{chain2} is argued identically, except
that we need to have
anticoncentration of the $|\inprod{a_i,Y}|$ in
intervals of size no smaller than $\rho = \eps/m$; this was
already shown to hold under $O(1/\rho^p)$-wise independence in
\cite[Lemma 2.5]{KNW10} for any $p$-stable distribution, and the Gaussian
is $p$-stable for $p=2$.

For the middle inequality we use Taylor's theorem, as was done
in \Lemma{kwise-fools}. If we truncate the Taylor polynomial at
degree-$(k-1)$ for $k$ even, then by our derivative bounds on mixed partials of $\tilde{I}_R^c$ from \Theorem{approx-uniform}, the error term is bounded by
$$ (2c)^k \cdot m^k\cdot \frac{\sum_{i=1}^m
  \E[\inprod{a_i,X}^k]}{k!} \le
(cm)^k\cdot 2^{O(k)}/k^{k/2},$$
with the inequality holding by \Lemma{deg1mom}, and the $m^k$ arising
as the analogue of the $4^k$ term that arose in
\Equation{taylor-process}. This is
at most $\eps$ for $k$ a sufficiently large constant times
$(cm)^2$, and
thus overall $k=\Omega(m^6/\eps^2)$-wise independence suffices.
\end{proof}

\begin{remark}
Several improvements are possible to reduce the dependence on $m$ in
\Theorem{intersection-halfspaces}. We presented the simplest proof we
are aware of which obtains a polynomial dependence on $m$, for clarity of
exposition.  See \Section{improved-halfspaces} for an improvement on
the dependence on $m$ to quartic.
\end{remark}

Our approach can also show that bounded independence fools
the intersection of any constant number $m$ of degree-$2$ threshold
functions.  Suppose the
degree-$2$ polynomials are $p_1,\ldots,p_m$.  Exactly as
in \Section{regular} we decompose each $p_i$ into
$p_{i,1}-p_{i,2}+p_{i,3}+p_{i,4} + C_i$. We then
define a region $R\subset \R^{4m}$ by $\{x : \forall i\in[m]\
x_{4i-3}^2 - x_{4i-2}^2 + x_{4i-1} + x_{4i} + C_i + \tr(A_{p_{i,3}}) >
0 \}$, and the map $F:\R^n\rightarrow \R^{4m}$ by
$$ F(x) = (M_{p_1}(X),\ldots,M_{p_n}(X))$$
for the map $M_p:\R^n\rightarrow\R^4$ defined in \Section{regular}.
The goal is then to show $\E[I_{R}(F(X))] \approx_\eps
\E[I_{R}(F(Y))]$, which is done identically as in the proof of
\Theorem{main-regular}.
We simply state the theorem here:

\begin{theorem}
Let $m>1$ be an integer.
Let $H_i=\{x:p_i(x)\ge 0\}$  for $i\in[m]$, for some
degree-$2$ polynomials $p_i:\R^n\rightarrow\R$.
Let $X$ be a vector of $n$
i.i.d. Gaussians, and $Y$ be a vector of $k$-wise independent
Gaussians with $k = \Omega(\poly(m)/ \eps^{8})$. Then,
$$ |\Pr[X\in \cap_{i=1}^mH_i] - \Pr[Y\in
  \cap_{i=1}^mH_i]| < \eps$$
\end{theorem}

Identical conclusions also hold for $X,Y$ being drawn from
$\{-1,1\}^n$, since we can apply the decision tree argument from
\Theorem{regularity-kwise} to each of
the $m$ polynomial threshold functions separately so that, by a union
bound, with probability at least $1-m\tau'$ each of the $m$ PTF
restrictions is either $\tau'$-close to a constant function, or is
$\tau'$-regular.  Thus for whatever setting of $\tau$ sufficed for the
case $m=1$ ($\tau = \eps^2$ for halfspaces \cite{DGJSV09} and $\tau=\eps^9$
for degree-$2$ threshold functions (\Theorem{main-regular})), we set
$\tau' = \tau/m$ then argue identically as before.
\section{Various Quantitative Improvements}\SectionName{quantitative}
In the main body of the paper, at various points we sacrificed proving
sharper bounds in exchange for clarity of exposition.  Here we discuss
various quantitative improvements that can be made in our arguments.

\subsection{Improved FT-mollification}\SectionName{improved-ftmol}
\ 
In \Theorem{approx-uniform}, we showed that for $F:\R^d\rightarrow\R$
bounded and $c>0$ arbitrary, $\|\partial^\beta \tilde{F}^c\|_\infty
\le \|F\|_\infty\cdot (2c)^{|\beta|}$ for all $\beta\in\N^d$.  We here
describe an improvement to this bound. The improvement comes by
sharpening our bound on $\|\partial^\beta B\|_1$.

We use the following fact, whose proof can be found in \cite{Folland01}.

\begin{fact}\FactName{sphere-trick}
For any multi-index $\alpha\in\N^d$,
$$ \int_{\|x\|_2 \le 1}x^{\alpha} dx = \begin{cases} 0 \ &
  \textrm{if some }\alpha_i\textrm{ is odd} \\ 
  \frac{2\prod_{i=1}^d\Gamma\left(\frac{\alpha_i+1}{2}\right)}{(|\alpha|
    + d)\cdot \Gamma\left(\frac{|\alpha|+d}{2}\right)}
  \ &\textrm{otherwise} \end{cases} .$$
\end{fact}

The following lemma is used in our sharpening of the upper bound on
$\|\partial^\beta B\|_1$.

\begin{lemma}\LemmaName{explicit-integral}
For a multi-index $\alpha\in\N^d$,
$$ \|x^{\alpha}\cdot b\|_2 \le \sqrt{\frac{\alpha! \cdot
  2^{O(|\alpha| + d)}}{(|\alpha| + d)^{|\alpha|}}}$$
\end{lemma}
\begin{proof}
By \Fact{sphere-trick},
\begin{eqnarray*}
\|x^{\alpha}\cdot b\|_2^2 &=& C_d\cdot \int_{\|x\|_2\le
  1}\left(x^{2\alpha} - 2\sum_{i=1}^d
x_i^2x^{2\alpha} + 2\sum_{i<j} x_i^2x_j^2x^{2\alpha} +
\sum_ix_i^4x^{2\alpha}\right)dx\\
&=& \frac{2C_d}{|\alpha|+d}\cdot \Bigg[\frac{\prod_{i=1}^d
  \Gamma\left(\alpha_i + \frac
    12\right)}{\Gamma\left(|\alpha| + \frac d2\right)} -
2\frac{\sum_{i=1}^d\left(\prod_{j\neq i} \Gamma\left(\alpha_j + \frac
    12\right)\right)\Gamma\left(\alpha_i + \frac
  32\right)}{\Gamma\left(|\alpha| + \frac d2 + \frac
  12\right)}\\
&&\hspace{.5in}{}+2\frac{\sum_{i<j}\left(\prod_{\substack{k\neq i\\k\neq j}}
    \Gamma\left(\alpha_k + \frac
    12\right)\right)\Gamma\left(\alpha_i + \frac
  32\right)\Gamma\left(\alpha_j + \frac
  32\right)}{\Gamma\left(|\alpha| + \frac d2 + \frac
  32\right)}\\
&&\hspace{.5in}{}+\frac{\sum_{i=1}^d\left(\prod_{j\neq i}
    \Gamma\left(\alpha_j + \frac
    12\right)\right)\Gamma\left(\alpha_i + \frac
  52\right)}{\Gamma\left(|\alpha| + \frac d2 + \frac
  32\right)}\Bigg] .
\end{eqnarray*}
Write
the above expression as 
\begin{equation}
\frac{2C_d}{|\alpha|+d}\cdot [W(\alpha) - X(\alpha) + Y(\alpha) +
Z(\alpha)] .
\end{equation}
For $\alpha = 0$ we have
$$W(0) = \frac{\pi^{d/2}}{\Gamma(\frac d2)},\hspace{.2in} X(0) =
d\cdot \frac{\pi^{d/2}}{\Gamma(\frac d2 + \frac 12)},\hspace{.2in} Y(0) =
d(d-1)\cdot \frac{\pi^{d/2}}{4\cdot \Gamma(\frac d2 + \frac
  32)},\hspace{.2in} Z(0) = d\cdot \frac{3\pi^{d/2}}{4\cdot \Gamma(\frac
  d2 +
  \frac 32)} .$$
Using the fact that $\Gamma(z+1) = z\Gamma(z)$, we can rewrite these
as
$$\frac{W(0)}{\pi^{d/2}} = \frac{1}{\Gamma(\frac d2)},\hspace{.2in}
\frac{X(0)}{\pi^{-d/2}} =
\frac{d}{\Gamma(\frac d2 + \frac 12)},\hspace{.2in}
\frac{Y(0)}{\pi^{-d/2}} =
\frac{d(d-1)}{2(d+1)\cdot \Gamma(\frac d2 + \frac
  12)},\hspace{.2in} \frac{Z(0)}{\pi^{-d/2}} = \frac{3d}{2(d+1)\cdot
  \Gamma(\frac
  d2 +
  \frac 12)} .$$
We thus have $W(0)-X(0)+Y(0)+Z(0) = \Omega(W(0)+Y(0)+Z(0))$. 
Since $2C_d(W(0)-X(0)+Y(0)+Z(0))/d = \|b\|_2^2 = 1$, it thus suffices to
show that $(W(\alpha)+Y(\alpha)+Z(\alpha))/(W(0)+Y(0)+Z(0)) \le
(\alpha! \cdot
  2^{O(|\alpha| + d)})\cdot(|\alpha| + d)^{-|\alpha|}$ for general
$\alpha$. This can be
seen just by showing the desired inequality for $W(\alpha)/W(0)$,
$Y(\alpha)/Y(0)$, and $Z(\alpha)/Z(0)$ separately. We do the
calculation for $W(\alpha)/W(0)$ here; the others are similar.

We have
$$W(0) \ge \frac{2^{-O(d)}}{d^{d/2}},\hspace{.3in} W(\alpha) \le
\frac{\alpha!\cdot 2^{O(|\alpha|+d)}}{(|\alpha| +
  d)^{|\alpha|+d/2}} ,$$
and thus
$$ \frac{W(\alpha)}{W(0)} \le
\frac{\alpha!\cdot d^{d/2} \cdot
  2^{O(|\alpha| + d)}}{(|\alpha| + d)^{|\alpha| + d/2}} \le
\frac{\alpha! \cdot
  2^{O(|\alpha| + d)}}{(|\alpha| + d)^{|\alpha|}} .$$
\end{proof}

\begin{lemma}\LemmaName{l1bound-better}
For any $\beta\in\N^d$ with $|\beta| = \Omega(d)$,
$\|\partial^\beta B\|_1 \le 2^{O(|\beta|)}\cdot \sqrt{\beta! \cdot
  |\beta|^{-|\beta|}}$.
\end{lemma}
\begin{proof}
The proof is nearly identical to the proof of \Lemma{l1bound}.  The
difference is in our bound of $\|x^\alpha\cdot b\|_2$.  In the proof
of \Lemma{l1bound}, we just used that $\|x^\alpha\cdot b\|_2 \le
\|b\|_2 = 1$.  However ,by
\Lemma{explicit-integral}, we can obtain the sharper bound
$$\|x^\alpha\cdot b\|_2 \le 2^{O(|\alpha|+d)}\sqrt{\alpha! \cdot
  (|\alpha|+d)^{-(|\alpha| + d)}} .$$
We then have
\begin{eqnarray*}
\|x^\alpha\cdot b\|_2\cdot \|x^{\beta-\alpha}\cdot b\|_2 &\le&
2^{O(|\beta|)}\sqrt{\alpha! \cdot
  (|\alpha|+d)^{-(|\alpha| + d)}\cdot (\beta-\alpha)! \cdot
  (|\beta-\alpha|+d)^{-(|\beta-\alpha| + d)}}\\
&\le& 2^{O(|\beta|)}\sqrt{\beta! \cdot
  |\beta|^{-|\beta|}}
\end{eqnarray*}
\end{proof}

We now have the following sharpening of item (i) from
\Theorem{approx-uniform}.  Over high dimension, for some $\beta$ the
improvement can be as large as a shrinking of our upper bound in
\Theorem{approx-uniform} by a $d^{-|\beta|/2}$ factor (for example,
when each $\beta_i$ is $|\beta|/d$).

\begin{theorem}\TheoremName{approx-uniform-improved}
Let $F:\R^d\rightarrow\R$ be bounded and $c>0$ be arbitrary, and
$\beta\in\N^d$ have $|\beta| = \Omega(d)$. Then,
$$\|\partial^\beta \tilde{F}^c\|_\infty \le \|F\|_\infty
  \cdot c^{|\beta|} \cdot 2^{O(|\beta|)}\cdot \sqrt{\beta! \cdot
  |\beta|^{-|\beta|}}$$
\end{theorem}
\begin{proof}
Note in \Equation{bl1-isit} in the proof of \Theorem{approx-uniform},
we showed that $\|\partial^\beta \tilde{F}^c\|_\infty \le \|F\|_\infty
\cdot c^{|\beta|} \cdot \|\partial^\beta B\|_1$.  The claim then
follows by
applying \Lemma{l1bound-better} to bound $\|\partial^\beta B\|_1$.
\end{proof}

\subsection{Improvements to fooling the intersection of
  halfspaces}\SectionName{improved-halfspaces}
\ 

In the proof of \Theorem{intersection-halfspaces}
in \Section{intersection-proofs}, we presented a proof showing that
$\Omega(m^6/\eps^2)$-independence $\eps$-fools the intersection of $m$
halfspaces under the Gaussian measure.  In fact, this dependence on
$m$ can be improved to quartic.  One factor of $m$ is shaved by using
the improved bound from \Theorem{approx-uniform-improved}, and another
factor of $m$ is shaved by a suitable change of basis.  The argument
used to shave the second factor of $m$ is specific to the Gaussian
case, and does not carry over to the Bernoulli setting.

\begin{theorem}\TheoremName{improved-halfspaces}
Let $m>1$ be an integer.
Let $H_i=\{x:\inprod{a_i,x}>\theta_i\}$ for $i\in[m]$, with
$\|a_i\|_2=1$ for all $i$.
Let $X$ be a vector of $n$
independent standard normals, and $Y$ be a vector of $k$-wise independent
Gaussians. Then for $k = \Omega(m^4/\eps^2)$ and even,
$$ |\Pr[X\in \cap_{i=1}^mH_i] - \Pr[Y\in \cap_{i=1}^m H_i]| < \eps$$
\end{theorem}
\begin{proof}
Let $v_1,\ldots,v_m\in\R^n$ be an orthonormal basis for a linear space
containing
the $a_i$. Define the region $R = \{x : \forall i\in[m]\ \sum_{j=1}^m
\inprod{a_i,v_j}x_j > \theta_i\}$ in $\R^m$. Note $R$ is
itself the intersection of $m$ halfspaces in $\R^m$, with the $i$th
halfspace having normal vector $b_i\in\R^m$ with $(b_i)_j =
\inprod{a_i, v_j}$.

We now define the map $F:\R^n\rightarrow\R^m$ by $F(x) =
(\inprod{x,v_1},\ldots, \inprod{x,v_m})$. It thus suffices to show
that $\E[I_R(F(X))] \approx_\eps \E[I_R(F(Y))]$. We do this by a chain
of inequalities, similarly as in the proof of
\Theorem{intersection-halfspaces}. Below we set $c = m^2/\eps$. 

\begin{equation}\EquationName{chain2-improved}
\E[I_{R}(F(X))] \approx_{\eps} \E[\tilde{I}_{R}^c(F(X))]
\approx_{\eps}  \E[\tilde{I}_{R}^c(F(Y))]  \approx_{\eps}
\E[I_{R}(F(Y))] .
\end{equation}

For the first inequality and last inequalities, since we performed an
orthonormal change of
basis the $F(X)_i$ remain independent standard normals, and we can
reuse the same analysis from the proof of
\Theorem{intersection-halfspaces} without modification.

For the middle inequality we use Taylor's theorem. If we set
$R(F(x)) = |\tilde{I}_R^c(F(x)) - P_{k-1}(F(x))|$ for $P_{k-1}$ the
degree-$(k-1)$ Taylor
polynomial approximating $\tilde{I}_R^c$, then 
\begin{eqnarray}
\nonumber R(x) &\le& \sum_{|\beta|=k}\|\partial^\beta
\tilde{I}_R^c\|_\infty\cdot
\frac{\prod_{i=1}^m |x_i|^{\beta_i}}{\beta!}\\
&\le&  \frac{2^{O(k)}\cdot c^k}{k^{k/2}}\cdot
\sum_{|\beta|=k} \frac{\prod_{i=1}^m
  |x_i|^{\beta_i}}{\sqrt{\beta!}} \EquationName{taylor-rocks}
\end{eqnarray}
by \Theorem{approx-uniform-improved}.
Now note
\begin{eqnarray}
\nonumber \sum_{|\beta|=k} \frac{\prod_{i=1}^m
    |x_i|^{\beta_i}}{\sqrt{\beta!}} &=&
\frac{1}{k!} \cdot \sum_{|\beta|=k} \sqrt{\beta!}\cdot
  \binom{k}{\beta} \cdot |x|^\beta \\
&\le& \frac{2^{O(k)}\cdot k^m}{k^k}\cdot
\sum_{\substack{|\beta|=k\\\forall i\ 2|\beta_i}} \sqrt{\beta!}
\cdot \binom{k}{\beta} \cdot |x|^{\beta}\EquationName{even-matching}\\
\nonumber &=& \frac{2^{O(k)}}{k^k}\cdot
\sum_{|\beta|=k/2}
\sqrt{(2\beta)!}\cdot \binom{k}{2\beta} \cdot |x|^{2\beta}\\
\nonumber &\le& \frac{2^{O(k)}}{k^{k/2}}\cdot
\sum_{|\beta|=k/2}
\binom{k/2}{\beta} \cdot |x|^{2\beta}\\
&=& \frac{2^{O(k)}}{k^{k/2}}\cdot \|x\|_2^k\EquationName{finale}
\end{eqnarray}
where $|x|$ denotes the vector
$(|x|_1,\ldots,|x|_m)$. \Equation{even-matching} holds for the
following reason.  Let
$\beta\in\N^m$ be arbitrary. Since $k$ is even, the number of odd
$\beta_i$ must be even.  Let $M$ be any perfect matching of the
indices $i$ with odd $\beta_i$. Then for $(i,j)\in M$, either
$|x_i|^{\beta_i+1}|x_j|^{\beta_j-1}$ or $|x_i|^{\beta_i-1}|x_j|^{\beta_j+1}$
must be at least as large as $|x_i|^{\beta_i}|x_j|^{\beta_j}$. Let
$\beta'$ be the new multi-index with only even indices obtained by
making all such replacements for $(i,j)\in M$.
We then
replace $\sqrt{\beta!}\cdot \binom{k}{\beta}\cdot |x|^\beta$ in the
summation
with $\sqrt{\beta'!}\cdot \binom{k}{\beta'}\cdot |x|^{\beta'}$. In doing
so, we have $x^{\beta'}\ge x^\beta$, but $\sqrt{\beta'!}\cdot
\binom{k}{\beta'}$ may have
decreased from $\sqrt{\beta!}\cdot \binom{k}{\beta}$, but by at most a
$2^{O(k)}k^m$ factor since each $\beta_i$ decreased by at most $1$ and
is at most $k$.  Also, in making all such replacements over all
$\beta\in\N^m$, we must now count each $\beta$ with even coordinates
at most $3^m$ times, since no such $\beta$ can be mapped to by more
than $3^m$ other multi-indices (if we replaced some multi-index with
$\beta$, that multi-index must have its $i$th coordinate either one
larger, one smaller, or
exactly equal to $\beta_i$ for each $i$). Note subsequent inequalities
dropped the $k^m$ term in the numerator since $2^{O(k)}\cdot k^m =
2^{O(k)}$ for our choice of $k$.

Now by \Equation{finale},
\begin{equation*}
 \E\left[\sum_{|\beta|=k} \frac{\prod_{i=1}^n
    |F(X)_i|^{\beta_i}}{\sqrt{\beta!}}\right] \le
\frac{2^{O(k)}}{k^{k/2}}\cdot \E\left[\left(\sum_{i=1}^m
    F(X)_i^2\right)^{k/2}\right] 
\end{equation*}
Since the $F(X)_i$ are
independent standard normal random variables,
$\sum_{i=1}^m F(X)_i^2$ follows a chi-squared distribution with $m$
degrees of freedom, and its $k/2$th moment is determined by $k$-wise
independence, and thus
\begin{equation}\EquationName{chi-moment}
 \E\left[\left(\sum_{i=1}^m
    F(X)_i^2\right)^{k/2}\right] = 2^{k/2}\cdot \frac{\Gamma(k/2 +
  m/2)}{\Gamma(m/2)} = 2^{O(k)}\cdot k^m \cdot k^{k/2} \le 2^{O(k)}
\cdot k^{k/2} .
\end{equation}
This finishes our proof, since by 
\Equation{taylor-rocks} the expected value of our Taylor error is
$$ \frac{2^{O(k)}\cdot c^k}{k^{k/2}}\cdot
\E\left[\sum_{|\beta|=k} \frac{\prod_{i=1}^m
  |F(X)_i|^{\beta_i}}{\sqrt{\beta!}}\right] = \frac{2^{O(k)}\cdot
c^k}{k^{k/2}} \cdot \left(\frac{2^{O(k)}}{k^{k/2}} \cdot 2^{O(k)}\cdot
k^{k/2}\right) = \frac{2^{O(k)}\cdot
c^k}{k^{k/2}},$$
which is $O(\eps)$ for $k = \Omega(c^2) =
\Omega(m^4/\eps^2)$.
\end{proof}

\end{document}